\documentclass[letter,11pt,reqno]{article}
\pdfoutput=1
\usepackage{color}
\newcommand{\citep}[1]{\cite{#1}}

\usepackage{mathptmx}
\usepackage{nicefrac}

\newcommand{\nfrac}{\nicefrac}

\usepackage{geometry, tabularx, graphicx}
\usepackage{amsmath, amsthm, amssymb, enumerate}
\usepackage{comment, fullpage}

\newtheorem{theorem}{Theorem}[section]
\newtheorem{claim}[theorem]{Claim}

\newtheorem{definition}[theorem]{Definition}

\newtheorem{observation}[theorem]{Observation}
\newtheorem{fact}[theorem]{Fact}
\newtheorem{lemma}[theorem]{Lemma}
\newtheorem{corollary}[theorem]{Corollary}

\newcommand{\defeq}{\stackrel{\textup{def}}{=}}
\newcommand{\eps}{\varepsilon}
\newcommand{\e}{\eps}

\newcommand{\Psymb}{\mathbb{P}}

\renewcommand{\Pr}{\ProbOp}
\renewcommand{\Pr}{\ProbOp}

\newcommand{\alert}[1]{}
\newcommand{\diff}[2]{\ensuremath{\frac{\text{d}#1}{\text{d}#2}}}
\newcommand{\pdiff}[2]{\ensuremath{\frac{\partial#1}{\partial#2}}}
\newcommand{\abs}[1]{\ensuremath{\left|#1\right|}}

\newcommand{\ceil}[1]{\ensuremath{\left\lceil#1\right\rceil}}

\newcommand{\Ex}[2]{\ensuremath{\mathbb{E}_{#1}\insq{#2}}}
\newcommand{\etal}[0]{\emph{et al}}

\newcommand{\ina}[1]{\ensuremath{\left<#1\right>}}
\newcommand{\inner}[2]{\ensuremath{\ina{#1,#2}}}
\newcommand{\inb}[1]{\ensuremath{\left\{#1\right\}}}
\newcommand{\inp}[1]{\ensuremath{\left(#1\right)}}
\newcommand{\insq}[1]{\ensuremath{\left[#1\right]}}
\newcommand{\M}[0]{\ensuremath{\mathcal{M}}}

\newcommand{\ones}[0]{\ensuremath{\mathbf{1}}}

\renewcommand{\Pr}[2][]{\ensuremath{\mathbb{P}_{#1}\insq{#2}}}
\newcommand{\norm}[2][]{\ensuremath{\left|\left|#2\right|\right|_{#1}}}

\newcommand{\ham}[1]{\ensuremath{d_H\inp{#1}}}
\newcommand{\dm}[1]{\ensuremath{d_{\rm match}\inp{#1}}}
\newcommand{\nt}[2]{\ensuremath{N^{#2}_{#1}}}
\newcommand{\dt}[2]{\ensuremath{D^{#2}_{#1}}}
\newcommand{\xt}[2]{\ensuremath{m^{#2}_{#1}}}
\newcommand{\phyp}[3]{\ensuremath{P^{\text{hyp}}\inp{#1\rightarrow#2;#3}}}
\newcommand{\pbin}[3]{\ensuremath{P^{\text{bin}}\inp{#1\rightarrow#2;#3}}}
\newcommand{\E}[1]{\ensuremath{\mathbb{E}\insq{#1}}}
\newcommand{\Var}[1]{\ensuremath{\mathbf{Var}\insq{#1}}}
\newcommand{\w}[1]{\ensuremath{w_H\inp{#1}}}
\renewcommand{\vec}[1]{\ensuremath{\mathbf{#1}}}
\renewcommand{\epsilon}[0]{\varepsilon}

\newcommand{\shrinkspace}[1]{}

\title{\bf A Finite Population Model of Molecular Evolution: \\ Theory and Computation}

\author{
    Narendra M. Dixit  \\
    {\small Department of Chemical Engineering} \\
    {\small Indian Institute of Science} \\
    {\small Bangalore-560 012, India }\\
    {\small Email: narendra@chemeng.iisc.ernet.in} 
  \and
    Piyush Srivastava\\
    {\small Computer Science Division}\\
    {\small University of California at Berkeley} \\
    {\small Berkeley - 94720, CA, USA}\\
    {\small Email: piyushsr@eecs.berkeley.edu}
  \and
    Nisheeth K. Vishnoi\\
    {\small Microsoft Research}  \\
    {\small Bangalore - 560 025, India}\\
    {\small Email: nisheeth.vishnoi@gmail.com}
}
\date{}
\begin{document}

\thispagestyle{empty}
\maketitle
\thispagestyle{empty}
\thispagestyle{empty}
\begin{abstract}

  This paper is concerned with the evolution of haploid organisms that
  reproduce asexually.  In a seminal piece of work, Eigen and
  coauthors proposed the quasispecies model in an
  attempt to understand such an evolutionary process.  Their  work
  has impacted antiviral treatment and vaccine design strategies.
Yet, predictions of the quasispecies model are at best viewed as a
guideline, primarily because it assumes an infinite population size,
whereas realistic population sizes can be quite small.  
In this paper we consider a population genetics-based model aimed at
understanding the evolution of such organisms with {\em finite}
population sizes and present a rigorous study
of the convergence and computational issues that arise therein.
Our first
result is structural and shows that, at any time during the evolution,
as the population size  tends to infinity, the distribution of genomes predicted by our
model converges to that predicted by the quasispecies model. This 
justifies the continued use of the quasispecies model to
derive guidelines for intervention.
While the stationary state in the quasispecies model is readily obtained, 
due to the explosion of the state space in our model, exact computations are prohibitive.
Our second set of results are computational in nature and address this
issue.  We derive conditions on the parameters of evolution under
which our stochastic model mixes rapidly. 
Further, for a class of widely used fitness landscapes we give a fast
deterministic algorithm which computes the stationary distribution of
our model.  These computational tools are expected to serve as a
framework for the modeling of strategies for the deployment of
mutagenic drugs.

\vspace{0.3in}

\noindent{\bf Topics: }Molecular Evolution, Quasispecies Theory.
\end{abstract}
\thispagestyle{empty}
\newpage
\thispagestyle{empty}
\tableofcontents

\thispagestyle{empty}
\newpage
\setcounter{page}{1}

\section{Introduction}

The rapid genomic evolution of viruses such as HIV has made the design of drugs and
vaccines with lasting activity one of the most difficult challenges of our time. A novel
intervention strategy which potentially outplays viruses in this evolutionary game was suggested by
the pioneering work of Eigen and coworkers \citep{eigen71,EMS89}. Eigen and coworkers considered the asexual evolution of a
haploid organism and found that when the mutation (or evolutionary) rate was small, the organism
survived as a collection of closely related yet distinct genomes together termed the quasispecies.
Viral populations in infected individuals are known to exist as such quasispecies \citep{LA10}. Remarkably, this
quasispecies model predicted that when the mutation rate increased beyond a critical value, called
the error threshold, the collection of genomes in the quasispecies ceased to be closely related; in
fact, a completely random collection of genomic sequences was predicted to emerge. This transition with increasing mutation rate
thus induced a severe loss of genetic information in the quasispecies and has been referred to as an
error catastrophe. The generic antiviral drug ribavirin has been shown to act as a \emph{mutagen}--an
agent that induces mutations--against poliovirus and trigger a severe loss of viral infectivity in
culture \citep{Crotty01}. This strategy of enhancing the viral mutation rate thus appears promising and particularly
advantageous because it is unlikely to be susceptible to failure through viral evolution-driven
development of drug resistance. Mutagenic drugs that attempt to induce an error catastrophe are 
thus being explored as a potential antiviral strategy
\citep{Crotty01,GP05,ADL04}, and one such drug for HIV is currently under clinical trials \citep{10.1371/journal.pone.0015135}.

The success of mutagenic strategies relies on accurate estimates of the error threshold of
the pathogens under consideration. Notwithstanding the tremendous insights into viral evolution the
quasispecies model provides, important gaps remain between the quasispecies model and the
realistic evolution of viruses and other haploid asexual organisms. First, whereas the model assumes
an infinite population size and, hence, adopts a deterministic approach, real populations are often
small enough to lend themselves to substantial stochastic effects. For instance, the effective
population size of HIV-1 in infected individuals is about $10^3-10^6$ \citep{Kouyos07,Balagam11}, which is thought to underlie the strongly
stochastic nature of HIV-1 evolution. Second, the model assumes a single-peak fitness landscape,
where one genomic sequence is assumed to be the fittest and all other genomes are equally less fit.
Realistic fitness landscapes are far more complex \citep{Hinkley}. There have been significant efforts in the last 30
years to close these gaps \citep{WilkeCommentary}. While more general landscapes have been successfully considered
in the quasispecies case \citep{saakian06}, a rigorous treatment of the finite population case has remained
elusive (see Section \ref{sec:disc-comp-with}). Importantly, it still
remains to be established whether the insights gained from the quasispecies model, such as the
occurrence of an error catastrophe, translate to the more realistic, finite population case.

Here, we consider a finite population model of the asexual evolution of a haploid organism.
Following standard population genetics-based descriptions \citep{citeulike:1993435}, the model considers evolution in
discrete, non-overlapping generations. Within each generation, genomes undergo reproduction (R),
selection (S), and mutation (M), yielding progeny genomes for the next generation. We analyze this
RSM model formally and establish the following key results. We show that in the infinite population
limit, the expected structure of the quasispecies predicted by the RSM model converges to that of
the quasispecies model. Thus, insights from the quasispecies model may be translated to the finite
population scenario. Indeed, we show further that the error threshold predicted by the RSM model
also converges in the infinite population limit to that of the quasispecies model. Finite population
models, such as the RSM model, appropriately tuned to mimic specific details, such as the fitness
landscape, of the pathogens under consideration may thus be employed to obtain realistic estimates of the error threshold. 

Unlike the quasispecies model, where the quasispecies is identified
readily using black-box eigenvector determination algorithms, identifying the expected quasispecies
of the RSM model is computationally prohibitive even for the smallest realistic genome and
population sizes. Monte Carlo sampling techniques are therefore often
resorted to~\citep{WilkeCommentary,althaus05,batorsky11,GD10,SimulationsPaper12}. Here, going
beyond the ideas of the quasispecies model, we examine the mixing properties of the RSM model.
We establish constraints on the model parameters under which the RSM model exhibits rapid mixing
and therefore allows fast estimation of the expected quasispecies structure. Finally, we suggest an
algorithm that uses the Markov Chain Monte Carlo paradigm to estimate the error threshold in the RSM model. 
Our study thus serves as a framework for elucidating
quantitative guidelines for the modeling of intervention strategies that employ mutagenic drugs.

The paper is organized as follows. In Section \ref{sec:prelimQM} we briefly describe the quasispecies model and the notion of the error threshold. In Section \ref{sec:finiteRSM} we setup the finite population RSM model, present our main results and outline the techniques employed.  In Section \ref{sec:disc-comp-with}
we discuss our results in the context of previous studies and highlight open problems arising from work. Formal statements of our results are presented in Section \ref{sec:form-stat-main}. Detailed proofs are contained in the Appendix.

\section{Preliminaries and Definitions}
\label{sec:prelimQM}

\subsection{Preliminaries}
We consider the evolution of a population of haploid organisms that
reproduce asexually.  In this evolutionary process the genome of each
individual is modeled as a string of nucleotides.  During
reproduction, the genome is copied with possible mutations, which can
be insertions, deletions, or substitutions.  In applications, such as
the modeling of viral evolution, it is often convenient to neglect
insertions, deletions, and substitutions other than transitions ($A$ to
$G$ or $C$ to $T$, and \textit{vice versa}).  Under this assumption, a
genome may, without loss of generality, be represented as a binary
string.  We thus represent a genome as an $L$-bit string
$\sigma=(\sigma_{1},\sigma_{2},\ldots, \sigma_{L} )$, where
$\sigma_{i} \in \{ 0,1\}.$ 

The \emph{fitness} of a genome is then modeled in terms of
its propensity to produce copies of itself.  
Specifically, the fitness of the genome $\sigma$ is
defined by the number of copies $a_\sigma$ of itself that it produces
in one round of replication (also called one generation).  However,
during replication, each bit of each of the $a_\sigma$ offsprings is
copied incorrectly with probability $\mu$ (called the error or
mutation \emph{rate}), thus potentially giving rise to an $L$-bit
string different from $\sigma$.  The fittest genome, also termed the
\emph{master sequence}, is without loss of generality assumed to be
${\bf 0}=(0,\ldots,0)$ so that $a_{\bf 0} > a_{\sigma}$ for all 
$\sigma \neq {\bf 0}$.  

The primary cause of the complexity and diversity in the
evolution of such organisms is the variety of possible fitness
landscapes, which \textit{a priori} can be arbitrary functions from
$\{0,1\}^{L}$ to the set of non-negative integers.  Several special classes of fitness
landscapes have been employed in the literature and we list the
important ones below. We will assume that $a_{\sigma} \geq 1.$ The case $a_{\sigma}=0$ for some $\sigma$'s has been
used~\citep{WK93,TH07,GD10}, and will be discussed in Section
\ref{sec:disc-comp-with}.  

\begin{enumerate}
\item ({\em General}) Here the only condition is that $a_{\sigma} \geq 1.$ 
\item ({\em Class Invariant~\citep{saakian06, TH07,PMD10,Balagam11}})
    In a class invariant landscape $a_{\sigma}$ depends only on the Hamming weight of $\sigma$.
\item ({\em Single Peak \citep{eigen71,NS89,PMD10}}) Here, we have $a_{\bf 0} > 1$ and $a_{\sigma}=1$ for
  all $\sigma \neq {\bf 0}$. 
\item ({\em Multiplicative \citep{TH07,WH96}}) These are parametrized by  $a_{1},\ldots,a_{L} \geq 1$ so that for a given $\sigma,$ $a_{\sigma} \defeq \prod_{i=1, \sigma_{i}=0}^{L}a_{i}$.
\end{enumerate}
\noindent
Other landscapes such as the simpler additive or linear landscapes and
more complex correlated landscapes have also been employed in the literature
\citep{BS93,Thomas97,nimwegen99}.

\subsection{The Quasispecies Model} 
\label{sec:quasispecies-model}
\noindent Eigen and coworkers \citep{eigen71,EMS89} gave the following differential equations for the
time-evolution of the fractional  population of the genome $\sigma$ at time $t,$ denoted by   $x_{\sigma}(t)$:
{$$ \diff{x_\sigma(t) }{t} =
  \sum_{\tau}x_{\tau}(t) a_\tau Q_{\tau\sigma} -
  x_\sigma (t)\bar{A}(t)\text{ for all $\sigma$}.$$ 
  \noindent
Here,  $Q_{\sigma\tau} \defeq \mu ^{d_H({\sigma,\tau} )} (1-\mu
)^{L-d_H({\sigma,\tau })} $ is the probability that $\sigma $ mutates
to $\tau $ and $d_H({\sigma,\tau} )$ is the Hamming distance between
$\sigma $ and $\tau.$ $\bar{A}(t)$ is the average fitness
$\sum_{\sigma}a_\sigma x_\sigma(t)$ at time $t$.  Defining
$A_{\sigma\tau} \defeq a_\sigma$ when $\sigma = \tau$ and $0$ otherwise, they showed that the vector of stationary frequencies,
${v}^\sigma_\mu \defeq \lim_{t \to \infty} x_{\sigma}(t) $, is
the dominant right eigenvector of the \emph{value matrix} $QA$ at mutation rate $\mu$ such that $\|{\bf v}_\mu\|_1 \defeq \sum_\sigma {v}_\mu^\sigma=1.$\footnote{Throughout this paper, we will be dealing with vectors over $\{0,1\}^L.$ Vectors will be typeset in boldface.  The components of a vector $\vec{x}$ will be denoted either by $x^\sigma$ or by $x_\sigma$ for $\sigma \in \{0,1\}^L$, based on convenience of notation in the context of use. } 
The collection of genomes determined by this dominant eigenvector,
which marks the culmination of the evolutionary process, is called the
{\em quasispecies}.  
It is important to note that the vector $\vec{v}_\mu$ is independent of the starting population distribution. 

We will mostly be concerned with the discrete
time version of the quasispecies model.  
In the discrete time case, $t=\{0,1,\ldots \},$  denoting the
fraction of genomes of type $\sigma$ at time $t$ by $m^\sigma_t,$
Eigen's equations can  be written as:
\begin{equation}\label{eq:conv-2}
  \xt{t+1}{\sigma} \defeq
  \frac{\sum_{\tau}\xt{t}{\tau}a_\tau Q_{\tau\sigma}}{\sum_{\tau}\xt{t}{\tau}a_\tau}.
\end{equation}
In vector notation, 
given the fractional population
$\vec{m}_t$ at
time $t$, the fractional population $\vec{m}_{t+1}$ at time $t+1$ 
is given by $\vec{m}_{t+1} = \vec{r}(\vec{m}_t)$,
where the $\sigma$ co-ordinate $r^\sigma$ of the vector valued
function $\vec{r}$ is defined as

\begin{equation}
\label{eq:r-def}
{r}^\sigma(\vec{x}) \defeq \frac{\sum_{\tau}a_\tau
  Q_{\tau\sigma}x_\tau}{\sum_{\tau}a_\tau x_\tau} =
\frac{\inp{QA\vec{x}}_\sigma}{\norm[1]{A\vec{x}}}\text{ and thus,
}\vec{r}(\vec{x}) = \frac{QA\vec{x}}{\norm[1]{A\vec{x}}} .
\end{equation}
Again, $\vec{m}_t$ can be shown to converge to $\vec{v}_\mu$ irrespective of the starting population distribution as $t$ goes to infinity. However, at any finite $t,$  $\vec{m}_t$ depends on the initial state $\vec{m}_0.$

\subsection{The Error Threshold}
With the single peak landscape, Eigen {\em et al.} observed empirically  that there
is a critical value $\mu_{c} \le 0.5$ for the mutation rate $\mu$ such
that for $\mu \ll \mu _{c} $, the quasispecies is dominated by the
master sequence, i.e. $v_\mu^{\bf 0} \ge v_\mu^{\sigma} \  \forall
\sigma$, whereas when $0.5 \ge \mu >\mu _{c} $ the quasispecies is
dispersed approximately uniformly.  The critical mutation rate $\mu
_{c} $ is called the \emph{error threshold} because the uniform
dispersal for $\mu >\mu _{c} $ implies a severe loss of representation
in the quasispecies of the genetic information encoded by the master
sequence. Evidently, this dispersal also decreases the mean fitness,
$\overline{A} \defeq \sum _{\sigma}a_{\sigma} v_\mu^{\sigma}
$. 

We note here that despite the notion of the error threshold being
widely recognized, no consensus exists on its definition; see
Wilke~\citep{WilkeCommentary}. Since, in most cases, $\vec{v}_\mu$ will never become exactly the uniform distribution on $\{0,1\}^L,$ it is clear that the goal is to find the smallest $\mu$ such that 
$\vec{v}_\mu$ is {\em close} to the uniform distribution on all genomes, i.e., the vector with every coordinate equal to $2^{-L}$, which we denote by  $\mathbf{\vec{U}}.$ To define the error threshold we also need a function that measures closeness: e.g. $\|\vec{v}_\mu-\vec{U}\|_1,$ $\|\vec{v}_\mu-\vec{U}\|_\infty$ or the difference in Shannon entropies of $\vec{v}_\mu$ and $\vec{U},$ namely $|\mathbf{H}(\vec{v}_\mu)-L|.$  Hence, for a given distance function $d$ which measures closeness of $\vec{v}_\mu$ and $\vec{U},$ and a bound on closeness $\eps>0,$ we define 
 $$ \mu_{c}^d(\eps) \defeq \min \{ \mu \in (0,1): d(\vec{v}_{\mu},\vec{U}) \leq \eps \}.$$

\noindent
First note that at $\mu=\nfrac{1}{2},$ the steady state vector $\vec{v}_\mu$ is {\em exactly} $\vec{U}.$ Hence, $\mu_c^d(\eps) \leq \nfrac{1}{2}$ for all $d,\eps >0.$
Second, note that  changing the distance function $d$ will change the error threshold quantitatively. Eigen and coworkers presented a heuristic argument that the error threshold should scale as $\nfrac{1}{L}$ for the single-peak model without any rigorous proofs of its existence and without mentioning any closeness function.

\section{A Finite Population Model and Our Main Results} 
\label{sec:finiteRSM}
In this section, we describe at an informal level the salient features
of our work, which comprises a finite population model to capture
molecular evolution, and our theoretical and computational results
associated with it. We give a high-level technical overview of the methods used
to prove our results in Section \ref{sec:overv-techn-contr}, while
precise definitions and formal statements of our results appear in
Section \ref{sec:form-stat-main}.  Proofs have been moved to the
Appendix due to considerations of space. 

\subsection{A Finite Population (RSM) Model}

We consider the following stochastic discrete time finite population model of
evolution which we call the RSM model.  The parameters are the same as in the quasispecies model: the genome length $L,$ the sequence space $\{0,1\}^L,$ a per bit mutation rate $\mu$ and the fitness landscape $\{a_\sigma\}_{\sigma \in \{0,1\}^L}$ with all $a_\sigma \geq 1$ and integers. At any time $t$, let
$N_{t}^{\sigma} $ be the number of genomes (a random variable)  of type $\sigma$, and fix
the total population $\sum _{\sigma }N_{t}^{\sigma}$ to be $N$.  This fixing of the population size to $N$ at each step is the key distinction from the quasispecies model and is a new parameter.  In
each time step, the ensuing evolution is then described in terms of the
following three steps.

\begin{enumerate}
\item ({\bf Reproduction}) First, in the reproduction
  step, each genome $\sigma$ produces $a_{\sigma} $ copies of
  itself, giving rise to an intermediate population $I_{t} \defeq \sum
  _{\sigma \in \{0,1\}^{L} }I_{t}^{\sigma} $, where $I_{t}^{\sigma}
  \defeq a_{\sigma} N_{t}^{\sigma} $.
\item ({\bf Selection}) Second, in the selection step, $N$ genomes are
  chosen at random without replacement from this intermediate
  population of size $I_{t}$, resulting in the selection of
  $S_{t}^{\sigma} $ genomes of type $\sigma$ where $S_{t}^{\sigma}
  \in \{ 0,1,\ldots, I_{t}^{\sigma}\}$ and $\sum _{\sigma \in \{0,1\}^{L}
  }S_{t}^{\sigma} =N \le I_{t}$.
\item ({\bf Mutation}) Third, in the mutation step, each selected
  genome is mutated with probability $\mu $ per bit, giving rise to
  the next generation of $N_{t+1}^\sigma$ genomes of type $\sigma$,
  such that $\sum_{\sigma\in\{0,1\}^L}N_{t+1}^\sigma = N$.
\end{enumerate}

\noindent
The starting state is denoted by $\vec{N}_0$ which is typically ggiven by  $N_0^\vec{0}=N$ and ${N}_0^\sigma=0$ for all $\sigma \neq  \vec{0},$ but we will often not use this assumption.
This RSM model is best viewed as a Markov chain where the state space is the set of functions $f: \{0,1\}^{2^L} \mapsto \{0,1,\ldots, N\}$ such that $\sum_\sigma f(\sigma)=N.$ 
Thus, the number of states of this Markov chain is ${{N+2^L-1}\choose{N}}$ which is roughly  $N^{2^L}.$ 
  It can be shown (see Fact \ref{fct:ergodic-chain}) that for any $0< \mu <1$  the transition matrix of this Markov chain has a unique stationary vector,   denoted by ${\pi}.$  $\pi$ is indexed by all $f$ satisfying the property above and $\sum_f \pi(f)=1,$ i.e., $\pi$ is a probability distribution over the state space of the RSM Markov chain.   

Let $\vec{D}_{t} \defeq
(N_{t}^{\sigma}/N)_{\sigma \in \{0,1\}^{L} } $ denote the random vector which captures the fractional population at time $t.$
It can also be shown that 
 $\lim_{t \rightarrow \infty} \E{\vec{D}_t|\vec{D}_0}$ exists and is independent of $\vec{D}_0.$ We denote this limit as $\E{\vec{D}_\infty}$ and it can be computed from the stationary vector $\pi$ as follows:
 $$ \E{D_\infty^\sigma} = \sum _{i=0}^N \sum _{f: f(\sigma)=i} \frac{i}{N} \cdot \pi(f).$$
\noindent
 Finally, we will subscript $\vec{D}_t$ with parameters such as $\mu,N$ when we want to highlight dependence on them, e.g. $\vec{D}_{t,\mu,N}.$
The key questions of interest, especially given the fact that computing $\pi$ would be prohibitive even for small values of $L$ and $N$, are:
\begin{enumerate} 
\item  For a fixed $t$, what does $\E{\vec{D}_t| \vec{D}_0}$ converge to as 
 $N$ increases? 
\item Is there a notion of error threshold in the RSM model?
\item   How to obtain an estimate of $\E{\vec{D}_\infty}$ efficiently? 
\end{enumerate}

\noindent
We present theoretical results that address all of the above questions. Note that if the answer to the first  question is that $\E{\vec{D}_t|\vec{D}_0}$ converges to the  prediction of the quasispecies model $\vec{m}_t$ with the same starting states ($\vec{m}_0=\vec{D}_0$), then it is important as we can leverage the significant understanding obtained from the study of the quasispecies model while incorporating stochastic effects with finite populations. 
For the second question, we need to first define a notion of the error threshold in the RSM model. We do so formally, given a distance function $d.$ 
\begin{definition} 
\label{def:error-thresholds-finite}
Let $\eps \geq 0.$  
$$\mu_{c}^d(\eps,N) \defeq \min \left\{ \mu \in (0,1):
    d(\E{\vec{D}_{\infty,\mu,N}},\vec{U}) \leq \eps 
  \right\},$$ where $\vec{U}$ is the uniform distribution over all genomes of length $L.$
\end{definition}
\noindent
What we want to understand is $\lim _{N \rightarrow \infty
}\mu_{c}^d(\eps,N).$ Again, if we can show that the answer here is
$\mu_c(\eps),$ then we can translate the insights from the
quasispecies model to the RSM model.

For the third question, first note that if we want to estimate the
error threshold, we need to be able to compute $\E{\vec{D}_\infty}.$
Secondly, we consider the standard notion of efficiency: the algorithm
to estimate $\E{\vec{D}_\infty}$ should be polynomial in the input
size. As we noted, the state space of the RSM Markov chain is
prohibitively large and computing the stationary state is
prohibitive. We employ the Markov Chain Monte Carlo method and run the
RSM process for some time $\tau$ such that it is guaranteed that
distribution from which $\vec{D}_\tau$ is drawn comes {\em
  statistically close} to the distribution from which $\vec{D}_\infty$
is drawn, irrespective of $\vec{D}_0.$ Simulating each step of the
random walk can be done efficiently. Hence, we are led to the question
of bounding the {\em mixing time} of the RSM Markov chain: the smallest time
the finite time distribution needs to come close to the steady state
distribution for all starting configurations.

The issue of how the input is presented is also important and we
briefly discuss it here. In one model, one can be given all $a_\sigma$
which would require bit length about $\sum_\sigma \log a_\sigma$ and
can, in principle, be as large or even larger than $2^L.$ Often, this
is not the case and either the values $a_\sigma$ are given by a simple
equation, or only some fixed number, say $k \ll 2^L$ of the values
$a_\sigma$ are strictly bigger than $1.$ In the latter case the input
has bit length roughly $O(k\log \max_\sigma a_\sigma +\log L+\log
\nfrac{1}{\mu}).$ Another case for the input is when $a_\sigma$ are
class invariant. In this case the input is of length roughly $O(L\log
\max_\sigma a_\sigma +\log \nfrac{1}{\mu}).$ We now proceed to
summarize our results.

\subsection{Our Results}

We now give informal statements of our main results, before
describing the mathematical techniques employed in the proofs of these
results.  The formal statements of the theorems described here appear
in Section \ref{sec:form-stat-main}, after a discussion in Section
\ref{sec:disc-comp-with}, while the formal proofs are deferred to the
Appendix.
\subsubsection{Convergence of the Quasispecies and the RSM Model} 
\begin{theorem}[Convergence of the RSM
 and Quasispecies Models]
\label{thm:convg-claim1}
Fix a fitness landscape $A$ with positive entries and a mutation
transition matrix $Q$.  Consider the RSM process started with the initial
state $\vec{D}_0$ and consider the evolution of the quasispecies model started
with the initial state $\vec{m}_0 = \vec{D}_0$.   Then for any fixed
time $t$,
  \begin{equation}
    \lim_{N\rightarrow\infty}\E{\vec{D}_{t}|\vec{D}_0} = \vec{m}_{t},  \end{equation}
      where $\vec{m}_{t}$ is the fractional population vector at time $t$ starting from $\vec{m}_0$ predicted by the quasispecies model.
\end{theorem}
The theorem shows that in the infinite population limit, the
stochastic fluctuations of the RSM process disappear, and the model
converges to the quasispecies model.  Informally, the main technical
difficulty in proving the above theorem is to establish a result of
the form $\lim _{N \to \infty } \mathbb{E}[\vec{D}_{t}|\vec{D}_{t-1}]
= \lim _{N\to \infty } \E{\vec{D}_t|\vec{D}_0}$ with probability $1$,
which would establish convergence to the quasispecies model.  The full
proof is deferred to the Appendix.  This convergence result allows us
to show that for any distance function $d$, a finitary version of
the error threshold, $\mu_c^d(\epsilon,N)$ as defined above,
converges to the error
threshold $\mu_c^d(\epsilon)$ of the quasispecies model, as the
population size goes to infinity.
These two results provide validation for the finite population RSM
model by establishing that in the infinite population limit, the
predictions from the RSM model converge to those of the quasispecies
model. We now move on to problems concerning the mixing time and
other computational issues of the RSM model.

\subsubsection{Computational Results in the  RSM 
 Model}  
As noted before, a  primary computational question in both the
quasispecies model and the RSM model is the determination of the
quasispecies, or the expected population profile at stationarity in
the RSM model, which can then be used to estimate error thresholds
(see Section \ref{sec:overv-techn-contr} for an overview and Sections \ref{sec:algor-estim-error-1} and
\ref{sec:algor-estim-error-2} for  details).  For the
quasispecies model, a satisfactory solution to this problem is
obtained via the observation that the quasispecies is the leading
right eigenvector of the $QA$ matrix.  The $QA$ matrix is of dimension
$2^L\times2^L$, and the above observation can thus be used to obtain
efficient algorithms using black-box eigenvector finding algorithms
for moderate values of $L$.  In the case of class-invariant fitness
landscapes, it is known~\citep{SS82} that one only needs to
find the leading eigenvector of an $(L+1)\times (L+1)$ matrix.

However, similar approaches are not as effective for the RSM model.
In this case, the stationary distribution is the leading eigenvector
of the transition matrix $\mathcal{M}$ of the RSM process which is of
dimension roughly $N^{2^L}$.  Using ideas similar to those referred to
above, one can reduce the running time for computing the stationary
distribution to $N^{O(L^2)}.$
\begin{theorem}[Computation of Steady State in the Class Invariant Case]
\label{thm:projected-chains1}
For any class invariant fitness landscape $A$, there is an algorithm
running in time $T= O(N^{O(L^2)})$ which computes the steady state of
the RSM process with population size $N$ and the genome length $L.$
\end{theorem}

 However, as mentioned before, in many applications, as in
the case of HIV, for instance, where $N\sim 10^{3} -10^{6} $ and
$L\sim 10^{4} $, the problem is still computationally prohibitive.  In
these cases, one typically resorts to Monte Carlo simulations of the
RSM process for estimating the population profile at
stationarity~\citep{SimulationsPaper12}, and thus we are led to
considering the mixing time of the RSM process.  
The following theorem derives conditions on the parameters of evolution  under
which the RSM model mixes rapidly.

\begin{theorem}[Mixing Time of the RSM Process]\label{thm:mixing-informal}
  Given a fitness landscape $A$, mutation rate $\mu$, the RSM process
  exhibits fast mixing if $(1-2\mu)\frac{\max_\tau a_\tau}{\min_\tau
    a_\tau}L  + \frac{1}{N} < 1.$
\end{theorem}

\noindent
Having stated our results, we now highlight the techniques employed in the proofs.

 \subsection{Overview of Our Technical Contributions}
\label{sec:overv-techn-contr}
As before, we will denote by $\vec{D}_t$ the random variable of
fractional populations after $t$ steps of the RSM process, and by $\vec{S}_t$
the random variable of the populations of genomes after the
replication and selection steps of the $(t+1)$-the step of the RSM
process.  

\paragraph{\bf Our convergence result (Theorem \ref{thm:convg-claim1}).}
The starting point of the proof of our convergence result is to observe that
$\E{\vec{D}_{t+1}|\vec{D}_t}$ has the same functional form $\vec{r}$
(as a function of $\vec{D}_t$) as the evolution equation of the discrete
time quasispecies model, with $\vec{r}$ as defined  in Equation \eqref{eq:r-def}: $\E{\vec{D}_{t+1} | \vec{D}_t} =
\vec{r}(\vec{D}_t)$.  Our high level approach is to first show that
$\vec{D}_{t+1}$ is actually \emph{concentrated} around
$\E{\vec{D}_{t+1}| \vec{D}_t}.$  Using the
Lipschitz continuity of the evolution function $\vec{r}$, we can then
chain these concentration results inductively to show that the
evolution of $\vec{D}_t$ is tightly concentrated around the evolution of the
discrete time quasispecies model, which allows us to show that $\E{\vec{D}_{t}}$ converges to the quasispecies as
$N\rightarrow\infty$.  To illustrate the ideas involved, we consider
the case $L =1$. Here the two genomes are $\{0,1\}.$  After the replication phase in the $t$-th step, there
are $a_0D^0_tN$ copies of ${0}$.  For the $i$-th copy, let
$R_i$ denote the indicator variable for this copy being selected in
the selection phase, so that $S_t^0 = \sum_{i=1}^{ a_0D_t^0N}{R_i}$.
Since the $R_i$'s are not independent, we cannot directly apply a
Chernoff bound.  However, since they are negatively correlated, one
expects concentration to hold, and this can indeed be shown using the
so-called method of bounded differences.  The same reasoning works for
$S_t^1$, and thus we get that given $\vec{D}_t$, the intermediate
population $\vec{S_t}$ after the replication and selection steps is
concentrated around its expectation with high probability.  We now
look at the mutation step.  Let $M_i$ be the indicator variable for
the $i$th genome being $\mathbf{0}$ after the mutation step.  We then
have $ND_{t+1}^0 = \sum_{i=1}^NM_i$.  Since the $M_i$'s are
independent random variables, it can be shown using a Chernoff bound
that given $\vec{S}_t$, ${D}_{t+1}^0$ is concentrated around
$\E{{D}_{t+1}^0|\vec{S}_t} =\nfrac{1}{N}( \mu S_t^0 + (1-\mu) S_t^1)$.
The two steps can then can be combined to show that given $\vec{D}_t$,
$D_{t+1}^0$ is concentrated around $\E{\E{{D}_{t+1}^0|
    \vec{S_t}}|\vec{D}_t}$ $=$ $\nfrac{1}{N}\E{\mu S_t^0 + (1-\mu)
  S_t^1|\vec{D}_t}$ $= \frac{(1-\mu)a_0D_t^0 + \mu D_t^1}{a_0D_t^0 +
  D_t^1}$.  The same reasoning works for $D_{t+1}^1$.

With some more work, this argument can be generalized to work for
arbitrary $L$.  The concentration guarantee we obtain is of the
following form: there are quantities $\epsilon_t$ and $p_t$ which are
both $o_N(1)$ such that given $\vec{D}_t$, $\abs{\vec{D}_{t+1} -
  \E{\vec{D}_{t+1}|\vec{D}_t}} \leq \epsilon_t$ with probability at
least $1 - p_t$.  In the next step, we chain these step-wise bounds
inductively in order to remove the conditioning and show that for all
$t\leq t_0$, $\vec{D}_{t}$ is concentrated around $\vec{m}_t$.  An
important component of the induction is the observation that $\vec{r}$
is Lipschitz continuous, which allows us to control the propagation of
the errors $\epsilon_t'$ in each step.  By the induction hypothesis,
we have that $\abs{\vec{D}_t - \vec{m}_t} \leq \epsilon_t'$ with
probability at least $1-p_t'$, where $\epsilon'_t$ and $p'_t$ are both
$o_N(1)$.  Assuming the Lipschitz constant of $\vec{r}$ to be $K$,
this implies that $\E{\vec{D}_{t+1}|\vec{D}_t} = \vec{r}(\vec{D}_t)$
is within distance $K\epsilon_t'$ of $\vec{m}_{t+1} =
\vec{r}(\vec{m}_t)$ with probability at least $1-p_t'$.  Applying the
convergence result from the first step, we then have that with
probability at least $ 1- p_{t+1}' = 1-p_t' - p_t$,
$\abs{\vec{D}_{t+1} - \vec{m}_{t+1}} \leq \epsilon_{t+1}' =
K\epsilon_t' + \epsilon_t$ of $\vec{m}_{t+1}$.  The quantities $p_t',
\epsilon_t'$ for $t \leq t_0$ can be chosen to be $o_N(1)$, which is
sufficient to show the required convergence.  The details appear in
Appendix \ref{sec:proof-theor-refthm:c}.  We now give an overview of
the proofs of our computational results.

\paragraph{\bf Computing the stationary distribution in the class invariant case (Theorem \ref{thm:projected-chains1}).}
Recall that the state space of the RSM Markov chain is roughly $N^{2^L}.$ However, if the fitness function is class invariant, we can show that the number of {\em distinct} coordinates in the state space is about $N^L.$  
To see this, first we define an equivalence relation on the  states of the RSM Markov chain. We say that   $f,g,$  which are functions from $\{0,1\}^L$ to $\{0,1,\ldots,N\}$ satisfying $\sum_\sigma f(\sigma)=N$ and $\sum_\sigma g(\sigma)=N,$ are \emph{equivalent} (denoted $f \equiv g$) if they  have the same statistics for every Hamming class, i.e., for every $0 \leq i \leq N,$
$$ \sum_{\inb{\sigma\in\inb{0,1}^L|\w\sigma = i}} f(\sigma)  = \sum_{\inb{\sigma\in\inb{0,1}^L|\w\sigma = i}} g(\sigma).$$  
Thus, the state space of the RSM Markov chain gets partitioned into about $(N+1)^{L+1}$ different classes.
Then, due to the fact that the fitness function is class invariant, it can be shown that  the transition probability of $f$ to any other equivalence class is the same as that of $g$ to the same class. Hence, one only needs to compute the transition probability from one equivalence class to another. This probability is a large binomial sum and one has to be careful in its computation and keep track of the number of bits required to represent each entry of this Markov chain over the equivalence classes. Once we have the transition matrix of this Markov chain, one can compute its largest eigenvector which corresponds to the stationary state. We show that, if one does this carefully, one can compute the eigenvector in time roughly $N^{O(L^2)}.$ The details appear in Section \ref{sec:proof-theor-refthm:p}.

 \paragraph{\bf Algorithm to compute the error threshold.}
Once we have the ability to either compute the stationary state of the
RSM process or derive independent samples from its  stationary state
(which allows us to estimate the relative frequencies of the genomes
at stationarity with a good precision by taking an average of the
sampled states), the algorithm to estimate  the error threshold is
simple. The idea is to start with a small value ($ \ll \nfrac{1}{L}$)
of $\mu,$ and to estimate/compute the stationary distribution of the RSM
process for the current value of $\mu.$ The algorithm then checks if the estimate of the 
stationary distribution is close to the uniform distribution on the
genomes in the measure of closeness of one's choice. If so, it stops and
outputs the current value of $\mu$ as an estimate of the error
threshold. Else, it increases $\mu$ by a very small amount and repeats the above steps. In
case direct computation of the stationary state of the RSM process is
computationally prohibitive, independent samples from the stationary
distribution of the RSM process are derived by simulating the RSM
process up to its mixing time. The number of samples required can be
estimated from a simple application of Chernoff bound on the random
variable corresponding to the stationary state distribution of the RSM
process. Hence, to establish bonds on the running time of the error threshold estimating algorithm, it is important to be able to bound the time it takes for the RSM process so that $\vec{D}_t$  comes close to the stationary state, $\vec{D}_\infty.$ Our next result is towards this. 

\paragraph{\bf Mixing time result (Theorem \ref{thm:mixing-informal}).} Since the stationary distribution
of the RSM chain is not very well understood, it is not clear how to
apply conductance-based geometric tools  or the canonical
paths method (see, for example, \citep{jerrum88}) in order to prove the mixing time result.
We are thus
led to more combinatorial coupling based methods.  Here one starts
with an integer valued metric $d$ on the state space of the Markov
chain, and then one runs two copies $X_t$ and $Y_t$ of the chain.  To
show fast mixing, it is then sufficient to prove that $X_t$ and $Y_t$
can be coupled so that $\E{d(X_{t+1},Y_{t+1})|(X_t,Y_t)} \leq \alpha <
1$. In general, defining a coupling can be tricky because one
needs to carefully argue that the marginals of the coupling agree with
the original Markov chain.

 We  define the coupling in two phases:
the first phase includes the replication and selection steps and the
second phase includes only the mutation step.  We begin with the
easier mutation step.  Let $I_t$ and $J_t$ denote the state of the two
RSM chains after the replication and selection steps.  For most
natural choices of the distance metric $d$, it is possible to couple
the mutation step using the standard coupling for the random walk on
the hypercube so that $\E{d(X_{t+1},Y_{t+1})|(I_t, J_t)} \leq
(1-2\mu)d(I_t, J_t)$.  The challenge however lies in controlling
$\E{d(I_t, J_t)|(X_t, Y_t)}$ while coupling the replication and
selection steps, since because of the global nature of the replication
and selection steps, $\E{d(I_t, J_t)|(X_t, Y_t)}$ can become quite
large.  We control this increase by a careful choice of the metric
$d$, and by appealing to the path coupling methods of Bubley and
Dyer~\citep{bubley97}.  The path coupling theorem says that for integer
valued $d$, it is sufficient to ensure $\E{d(X_{t+1},Y_{t+1})|(X_t,
  Y_t)} \leq \alpha \leq 1$ \emph{only} for states $X_t$ and $Y_t$
satisfying $d(X_t, Y_t) = 1$ in order to establish fast mixing. Our coupling is then defined as follows.
Fix a permutation of the $N$ genomes in the chain $X_t$: $d(X_t, Y_t)$
is then the minimum over all possible permutations of the $N$ genomes
in $Y_t$ of the sum of the Hamming distances between the genomes at the
same positions in the two permutations.  The main technical step is to
show that for this $d$, the replication and selection steps can be
coupled in such a way that starting from $X_t$ and $Y_t$ satisfying
$d(X_t, Y_t) = 1$, $\E{d(I_t, J_t) | (X_t, Y_t)}$ after these steps is
at most $\frac{N}{N-1}\frac{\max_\sigma a_\sigma}{\min_\tau a_\tau}L$.
It is in this step that we use the form of the distance metric $d$
crucially; the details of the coupling are somewhat technical and
involve arguing carefully that the coupling is valid, and are given in
Appendix \ref{sec:proof-theor-refthm:m}.  We then combine this with
the coupling for the mutation step described above to show contraction
in the expected distance under the condition
$(1-2\mu)\frac{N}{N-1}\frac{\max_\sigma a_\sigma}{\min_\tau a_\tau}L <
1$.

\section{Discussion and Future Perspectives}
\label{sec:disc-comp-with}

\subsection{Previous Work}
The notion of the quasispecies and the existence
of an error threshold were recognized first by Eigen and his coworkers
in the 1970s and 1980s \citep{eigen71,EMS89}.  Translation of these ideas into intervention strategies requires overcoming two key limitations of the quasispecies model. First,
the model  assumes an infinite population size, whereas
realistic population sizes can be quite small. With HIV, for instance,
the effective population size is estimated to be $\sim 10^3 - 10^6$
\citep{Kouyos07,Balagam11}. Second, the theory assumes a single-peak
fitness landscape, whereas realistic landscapes can be far more complex
\citep{Bonhoeffer04,Hinkley}.
Efforts over the
last several decades have attempted to overcome these limitations of the quasispecies model 
\citep{NS89,BS93,WH96,Thomas97,alves98,saakian06,TH07,PMD10} (see Wilke
\citep{WilkeCommentary} for a recent review). The finite population
case, however, has remained difficult to solve in full generality.
Most studies resort to simulations or
use approximate or heuristic approaches to describe the finite
population case, and we discuss some of these here.  

Nowak and Schuster~\citep{NS89} used a birth-death process to model the
underlying evolution in finite populations and using simulations
predicted that the error threshold scales as $1/\sqrt{N}.$ Their
model, however, does not converge to the quasispecies model as $N$ goes to
infinity. Alves and Fontanari~\citep{alves98} present a model which
employs a two-stage sampling with replacement in the selection
process: first sampling uniformly from the population, and then
sampling from the obtained sample with biases proportional to the
fitness. They note, however, that sampling with replacement destroys
the negative correlation between the selection of two individuals of
the same species induced by the finite population constraint when
selection is implemented using sampling \emph{without} replacement.
They find that the error threshold scales as $1/N.$ Further, they only
analyze a heuristic deterministic approximation to their model, and do
not consider rigorously the question of convergence of their original
model to the quasispecies model. The closest to our convergence result is that by  van Nimwegen
\etal.~\citep{nimwegen99} who show convergence to a deterministic
Eigen-like dynamics but again employ sampling with replacement and use
special cases of additive fitness landscapes. With finite populations,
there can be a significant statistical difference between sampling with
and without replacement, the latter (which we employ) being more
realistic. It is well known that as $N \rightarrow \infty$ the
difference between sampling with and without replacement shrinks, but
then as we prove, so does the difference between our population
genetics model and the quasispecies model. Their convergence proof
has a similar structure as ours but we are able to use Chernoff bound
type inequalities which are much stronger than the second moment
inequalities used by them. Consequently, our  convergence results are
quantitatively stronger.  We additionally prove convergence of the
error threshold and fast mixing, questions not considered by
\citep{nimwegen99}. More recently,
Musso~\citep{musso11:} presented the transition matrix for sampling
with replacement in the case $L=1$ and also claimed convergence to
the quasispecies model in the deterministic limit.  No attempt, however,
is made in~\citep{musso11:} to make this latter claim rigorous.
Another class of studies relies on approximations and heuristics
inspired from physics, and in particular statistical mechanics, to
render the finite population case mathematically tractable (e.g.,
\citep{BK98,Saakian08,PMD10}).

While previous studies have focused extensively on
the fractional distribution of genomes at stationarity, little is known of the time to reach the
stationary state. Campos and Fontanari \citep{CF99}  show that in the limit of infinitely large genome lengths ($L \rightarrow \infty$)  and population sizes ($N \rightarrow \infty$) and with the single peak fitness  
landscape, the timescale associated with the decline of the master sequence is 
$\nfrac{1}{\ln (qa)}$ where $q$ 
is the probability that a genome is replicated without error, and $a$ 
is the relative fitness
of the master sequence. Further, they show that with finite populations, this timescale is proportional to $\sqrt{N}.$  The mixing time when $L$ and $N$ are both finite and when the fitness landscape
is more general than the single peak remains unknown. The latter mixing time has practical
significance in the modeling of the action of mutagenic drugs, as it respresents the duration of therapy required 
to ensure completion of the transition to the error catastrophe. Our study presents
conditions when the mixing is rapid and hence the transition to error catastrophe occurs quickly.
Further, for computational studies that attempt to realize this transition \emph{in silico}, our study presents
an algorithm that allows efficient Monte Carlo sampling-based estimation of the error threshold.

\subsection{Applications of the RSM Model}
\label{sec:appl-finite-popul}

The motivation behind the RSM model and the algorithms
discussed here is to get a basic framework for understanding the
evolution of viruses of current interest such as HIV. Making concrete
predictions relevant to the clinical setting requires super-imposing
the specifics of the viruses of concern on the present framework. This
often involves subtle modifications of the RSM process along with
validation against data.  For example, in related recent work, two of the authors and
their co-workers applied the RSM model to mimic the within-host
genomic evolution of HIV-1~\citep{SimulationsPaper12}.
It has been shown before ~\citep{Balagam11} that these simulations
quantitatively capture data of the evolution of viral genomic
diversity in patients over extended durations ($\sim 10$ years)
following infection and the approach is extended in
\citep{SimulationsPaper12} to estimate the error threshold of
HIV-1.   We envision that similar adaptation of our model will prove useful in
elucidating the evolution and treatment guidelines for other asexual
haploid organisms of interest. 

\subsection{Critique of the RSM Model}
We note that our structural and computational results are independent of the nature of the
fitness landscape so long as there are no lethal mutations ($a_\sigma \neq 0$ for any $\sigma$). 
Our model, however, does not consider lethal
mutations. While letting some $a_{\sigma}$ be $0$ does not affect the
quasispecies model due to the constant rescaling involved, it
introduces an absorbing state in the RSM Markov chain, thus making it
non-ergodic, and causing the population to eventually decrease to zero.  While
Wilke and others \citep{WilkeCommentary, WK93, TH07} have commented on
the role of lethal mutations in extreme cases, establishing their full
implications lies beyond the scope of the present paper.
Although lethal mutations do occur, it turns out that in many important scenarios such as the evolution of HIV-1, a Hamming class invariant landscape without lethal mutations appears to capture key features of the underlying
fitness interactions \citep{Bonhoeffer04}, rendering our RSM framework applicable. 

Finally, we note that our assumption of a fixed population size, $N,$ is consistent with the
widely accepted population genetics-based models of evolution, where a constant effective
population size is employed to quantify the strength of stochastic effects \citep{citeulike:1993435}. Note that allowing $N$ to
vary with time (generations), does not increase the complexity in our model. The distinction
between an infinite population model and a finite population model arises from the culling of the
population in the latter model in order to maintain a finite population size. A fixed $N$ or varying $N$
will only result in different extents of culling in different generations, but will not change the overall
structure of the model. The advantage in keeping $N$ constant for our present study is that it allows
easier examination of the convergence to the quasispecies model.

\subsection{Open Problems}
Our study of the quasispecies and RSM models has
revealed several interesting and important problems. We list the main
ones here.

\paragraph{\bf Structure of the Quasispecies.}
Perhaps the most attractive feature of finite population models as
opposed to the quasispecies model is that they can be used to study
the effect of random genetic drift on inter-patient variations.
Inter-patient variations in disease progression and response to
treatments are known to be significant with HIV infection
\citep{Nijhuis98,Gonzalez05}. The collection of viral particles in an
infected individual may be thought of as one realization of the random
viral evolutionary process, and $\lim_{t \to \infty} \Var
{D_{t}^{\sigma}} $ then provides an estimate of inter-patient
variations in viral evolution due to the effect of the finite
population size.  Thus, in addition to the structure of the
quasispecies in the finite population case, defined by the expected
frequencies $\lim_{t \to \infty}\E{\vec{D}_t}$ when $N<\infty
$, the variance of the frequencies $\lim_{t \to \infty}\Var{D_{t }^{\sigma}
}$ as a function of the population size $N$ is also an important quantity to be studied.

\paragraph{\bf Error Threshold.} In the quasispecies model with the
single-peak fitness landscape, $\mu _{c} $ has been found, without a rigorous proof,  to be
$O(1/L)$, so that an error catastrophe occurs for $\mu \ll 0.5$ (e.g.,
see \citep{EMS89}). Further, the transition is sharp, so that a small
increase in $\mu $ from below to above $\mu _{c} $ induces a dramatic
change in the quasispecies structure. With other fitness landscapes,
such as the multiplicative landscape, however, the quasispecies
approaches the uniform distribution gradually as $\mu $ approaches
$0.5$ \citep{WH96}. Further, lethal mutations, where $a_{\sigma}=0$ for
some $\sigma$'s, appear to show the existence of an error threshold
only if multiple mutations in a single replication are allowed
\citep{WK93,TH07,WilkeCommentary}. Thus, the conditions under which a
sharp transition leading to an error catastrophe at $\mu _{c} \ll 0.5$
would occur remain to be established.  Second, the dependence of $\mu
_{c} $ on $N$ remains to be identified. While some simulations suggest
a $1/\sqrt{N}$ dependence \citep{NS89,BS93}, others find the dependence
to go as $1/N$ \citep{alves98}. As pointed out before, knowledge of
$\mu _{c} $ for finite $N$ is important in the modeling of
antiviral strategies based on mutagenic drugs.

\paragraph{\bf Mixing Time.}
The main outstanding question here is to get a tight
bound on the mixing time of the RSM Markov chain for a full range of evolutionary parameters. 
We notice that our result shows a good mixing
time bound only under certain conditions on the parameters.  Though we
conjecture that the chain is rapidly mixing for other values of the
parameters too, we believe that  novel methods would be needed to
extend our results in this direction.
 Apart from being useful in determining the time required for simulations to produce samples from the stationary distribution, the mixing time bounds  also  have biological
significance.  For example, when modeling the effect of a mutagenic
drug under the RSM model, the convergence rate would
models the minimum required duration of treatment before the error catastrophe occurs. 

\section{Formal Statements of Main Results}
\label{sec:form-stat-main}
\subsection{Preliminaries and Definitions}
\label{sec:prel-defin}
In this section we present rigorous statements of our results. Several definitions may be found repeated here in the interest of the readability of this section.
We recall that genomes of length $L$ are denoted by
$L$-bit 0-1 strings.  We will denote the Hamming distance between
genomes $\sigma$ and $\tau$ by $\ham{\sigma, \tau}$, and the Hamming
weight of a genome $\sigma$ by $\w{\sigma}$.  A \emph{population} is
defined as a multiset of genomes of the same length.  While discussing
the RSM model, we will fix the size of the population to be $N$.
\paragraph{\bf The Markov Chain for the RSM Model.}
We will denote the evolution of the RSM
process using a time-indexed sequence of vector valued random
variables $\inp{\vec{N}_t}_{t=0}^\infty$. The entries of $\vec{N}_t$
are indexed by genomes $\sigma$, and the entry
$\nt{t}{\sigma}$ denotes the number of genomes of type $\sigma$ at
time $t$.  At every time $t$, 
$\sum_{\sigma\in \inb{0,1}^L}\nt{t}{\sigma} =
N$.  The random variables $\dt{t}{\sigma} \defeq \nfrac{\nt{t}{\sigma}}{N}$
denote the \emph{fractional population} of the genome $\sigma$ at time
$t$.
\begin{description}
\item[Reproduction Step and the Fitness Landscape]
  In the \emph{reproduction} step, each genome $\sigma$ produces
  $a_{\sigma}$ copies of itself, so that the number of genomes of type
  $\sigma$ after this step is $I_{t}^\sigma \defeq
  a_\sigma\nt{t}{\sigma}$, and the total number of genomes is
  $I_t\defeq \sum_{\sigma}a_{\sigma}\nt{t}{\sigma}$.  
  The
matrix $A$ defined by $A_{\sigma\sigma} \defeq a_\sigma$ and $A_{\sigma\tau} \defeq0$ for $\sigma \neq \tau$  is called the
  \emph{fitness landscape}.  
  {
  The fitness landscape is said to be
  \emph{class-invariant} if $a_\sigma$ depends only on the Hamming
  weight  of $\sigma$.}
  By a slight abuse of notation, we
  will denote by $a_i$ the fitness of all genomes with Hamming weight
  $i$ in the class invariant case.

\item[Selection and Mutation Steps and the Mutation Rate]
  In the \emph{selection} step, $N$ genomes are sampled 
  without replacement from the  genomes obtained
  after the reproduction step.  In the \emph{mutation} step, each bit
  of each of the $N$ genomes obtained after the selection step is
  flipped with a probability $\mu $, called the
  \emph{mutation rate}.  The \emph{mutation transition matrix} $Q$
  defined by $Q_{\sigma\tau}  \defeq
  \mu^{\ham{\sigma,\tau}}(1-\mu)^{L-\ham{\sigma,\tau}}$ gives the
  probability that a genome of type $\sigma$ mutates to one of type
  $\tau$ in the mutation step.
\end{description}
\noindent The RSM process as described above is a Markov chain on the
state space of functions $f: \inb{0,1}^{L}\longrightarrow
\mathbb{N}$, satisfying $\sum_{\sigma\in\inb{0,1}^{L}}f(\sigma) = N$.
The  transition matrix $\mathcal{M}$ of this chain is
described  in
Section \ref{sec:start-state-trans}.

\begin{fact}
  \label{fct:ergodic-chain}
  When the mutation rate $\mu \in (0,1)$ and $a_{\sigma} >0$ for all
  $\sigma,$ the Markov chain $\mathcal{M}$ corresponding to the RSM
  process is ergodic, and hence has a unique stationary distribution.
\end{fact}

\noindent
This is a simple consequence of the fact that $\mu$ and $A$ are
positive.  See Section \ref{sec:proofs-omitted-from-1} for a proof.

\paragraph{\bf Important Statistics of the RSM Process and the Projected
    RSM Process.}
The transition matrix $\M$ of the RSM Markov chain is of
dimension $\binom{N+2^L-1}{N} \times \binom{N+2^L-1}{N}$.  When the fitness landscape $A$ is class invariant, one can
get a projected Markov chain with a significantly smaller state space
which can still be used to compute the average fitness and the average
population of each fitness class at stationarity.  Consider
equivalence classes indexed by functions $h: [0,L] \longrightarrow N$ with
$\sum_{i=0}^L h(i) = N$, such that a function $f$ in the state space
of $\mathcal{M}$ is in the equivalence class $[h]$ if and only if for
every $i \in [0,L]$, $
\sum_{\inb{\sigma\in\inb{0,1}^L|\w\sigma = i}} f(\sigma) = h(i).
$
We then have the following lemma the proof of which is in
Section \ref{sec:proofs-omitted-from-1}.
\begin{lemma}
  \label{lem:proj-strong}
  Let $f, g$ belong to the same equivalence class $h$ as defined
  above, and let $[h']$ be another equivalence class.  We then have
  $\mathcal{M}(f,[h']) = \mathcal{M}(g, [h'])$.  
\end{lemma}

\noindent Thus, we can consider the projected Markov chain $\M_w$ with
state space 
{$$\Omega_w = \inb{[h]\left| \sum_{i=0}^Lh(i) =
    N\right.}.$$}
Notice that $|\Omega_w| = \binom{N + L }{L}$.
Also, if $\pi_w$ is the stationary distribution of $\M_w$, then by the
projection property
$$
\pi_w([h]) = \sum_{f \in [h]}\pi(f).
$$
This property implies that the expected populations for every Hamming
class of genomes and, hence, the expected average fitness at
stationarity, are the same for $\M_w$ and $\M$.

\paragraph{\bf Mixing Time.}
We will denote by $\pi$ the stationary distribution of the RSM
process, and let $\vec{N}_{\infty}$ be a random variable distributed
according to $\pi$.  We know that the distributions of the random
variables $\vec{N}_t$ converge in total variation distance (and hence
in distribution) to $\pi$, due to the ergodicity of the RSM process.
We fix our notation for mixing times in this section.  

 \begin{definition}
  The total variation distance between two probability distributions
  $\mathcal{D}_{1}$ and $\mathcal{D}_{2}$ on the sample space $\Omega$
  is defined by $\|\mathcal{D}_{1} - \mathcal{D}_{2} \|_{TV} = \max_{A
    \subseteq \Omega} |\mathcal{D}_{1}(A) - \mathcal{D}_{2}(A)|.$
\end{definition}

\begin{definition}
  Fix a Markov chain $\mathcal{N}$ on a state space $S$.  We define
{$$ d(t) \defeq  \max_{\alpha \in S} \|
   \mathcal{N}^{t}(\alpha,\cdot) - \pi \|_{TV}.$$}
For $0\leq \epsilon\leq 1/2$, the mixing time of $\mathcal{N}$ is
defined by 
{$$ \tau_{\rm mix}(\e) \defeq \min \{t : d(t) \leq \eps
\}.$$}
\end{definition}

 \noindent Notice that by the projection property, the mixing time of the
projected RSM chain $\mathcal{M}_w$ is at most the mixing time of the
original RSM chain $\mathcal{M}$.

\paragraph{\bf Error Thresholds.}
In the following definition, we specifically
emphasize the dependence of the random variables
$\vec{N}_{t},\vec{D}_{t}$ on $\mu,N$ by denoting them as
$\vec{N}_{t,\mu,N}$ and $\vec{D}_{t,\mu,N}$.  We denote by
$\vec{D}_{\infty,\mu,N}$ a version of $\vec{D}_{t,\mu,N}$  distributed according to the
stationary distribution of the RSM process, and by $\vec{U}$ the
uniform distribution over genomes.  
Given a distance function $d,$ one can define the error threshold with respect to $d$ as follows. 
\begin{definition} [Error Threshold for the RSM Model] 

\label{def:error-thresholds-finite}
Let $\eps \geq 0.$  
$$\mu_{c}^d(\eps,N) \defeq \min \left\{ \mu \in (0,1):
    d(\E{\vec{D}_{\infty,\mu,N}},\vec{U}) \leq \eps 
  \right\},$$ where $\vec{U}$ is the uniform distribution over all genomes of length $L.$
\end{definition}

\noindent
Of particular interest is the function $d^h:$ for any two
distributions $\mathcal{D}_1$ and $\mathcal{D}_2$ over genomes,
$d^h(\mathcal{D}_1, \mathcal{D}_2)$ denotes
$\abs{\sum_{\sigma}w(\sigma)(\mathcal{D}_1(\sigma) -
  \mathcal{D}_2(\sigma))}$. The corresponding $\mu$ will carry the superscript $h.$

\subsection{Convergence to the Quasispecies Model}
\label{sec:conv-eigens-model}
Our first main result is that the RSM model converges to the quasispecies model.
\begin{theorem}
\label{thm:convg-claim}
Fix a fitness landscape $A$ with positive entries and a mutation
transition matrix $Q$.  Consider the RSM started with the initial
state $\vec{D}_0$ and consider the evolution of the quasispecies model started
with the initial state $\vec{m}_0 = \vec{D}_0$.   Then for any fixed
time $t_0$,
{\begin{equation}
    \lim_{N\rightarrow\infty}\E{\vec{D}_{t_0}|\vec{D}_0} =
    \vec{m}_{t_0},  
\end{equation}}
  where $\vec{m}_{t_0}$ is the state of evolution of the
  quasispecies model at time $t_0$ starting from $\vec{m}_0.$
\end{theorem}
\noindent
The proof of the above theorem is relegated to Section
\ref{sec:proof-theor-refthm:c}.
As a corollary to the theorem above, one can show that there is  convergence of a finitary version of  the error threshold
$\mu_c^h(\epsilon,N)$  to the error threshold  $\mu_c^h(\epsilon)$ for
the quasispecies model, as the population size goes to infinity.
Formally, we have the following:
\begin{corollary}
\label{coroll:conv-error-threshold}
Fix a mutation rate $\mu \leq 1/2$ and an error parameter $\epsilon$.
For every $\delta > 0$, there exists a time $t_0 > 0$ such that for $t
> t_0$, one can find an $N_{\delta,t}$ such that for $N >
N_{\delta,t}$,
{\begin{eqnarray}
    d^h\inp{\E{\vec{D}_{t,\mu,N}}, \vec{U}} \geq
    \epsilon - \delta,&\text{ when $\mu <
      \mu_c^h(\epsilon)$, and}\nonumber\\
    d^h\inp{\E{\vec{D}_{t,\mu,N}}, \vec{U}} \leq
    \epsilon + \delta&\text{ when $\mu =
      \mu_c^h(\epsilon)$,}\nonumber
  \end{eqnarray}}
Here we use the subscripts
$\mu$ and $N$ to emphasize the dependence of the distribution of
$\vec{D}_t$ on $\mu$ and $N$.
\end{corollary}
\noindent Although we will prove our results for the error threshold  in terms
of the average Hamming distance, it is easy to translate our results
to other common dispersal measures as described in Section
\ref{sec:relat-betw-error}.  
\noindent The proof of the above Corollary follows easily from Theorem
\ref{thm:convg-claim} and is given in Section
\ref{sec:error-threshold-convergence}.  We note here that extending
the above corollary to get convergence of finite population error
thresholds depends upon proving a strengthened version of our
convergence result (Theorem \ref{thm:convg-claim}), which we leave as
an open problem. 
In fact, on the basis of simulation results, we
conjecture that for fixed $\epsilon$, $\mu_c^h(\epsilon,N)$
monotonically increases to $\mu_c^h(\epsilon)$.

\subsection{Computational Results}
\label{sec:algor-estim-error}
\subsubsection{Mixing Time Bounds on the RSM Process}
We give a coupling argument in Section  \ref{sec:proof-theor-refthm:m}
which allows us to prove the following theorem.
\begin{theorem}
\label{thm:mixing-main}
  Fix  $0< \mu \leq \nfrac{1}{2}$, and a fitness landscape $A$.  Let
{$$
K(A, \mu) \defeq (1-2\mu) \frac{N}{N - 1}\frac{\max_\sigma
  {a_\sigma}}{\min_\tau a_\tau}L. 
$$}
When $K(A, \mu) < 1$,  we have $\tau_{\rm mix}(\epsilon) =
O\inp{\frac{\log(NL/\epsilon)}{\log (1/K)}}$.  
\end{theorem}

\subsubsection{Computing the Stationary Distribution in the Class Invariant Case}
\label{sec:algor-estim-error-1}
\begin{theorem}
\label{thm:projected-chains}
For every $A$ which is class invariant, and $\mu$ which can be
represented using $b$ bits, there is an algorithm running in time $T$
given by 
$$T \defeq \tilde{O}\inp{bL^3 + \inp{\frac{(N+L)^L}{L!}}^{3}(NbL + L^2 + N\max a_\sigma) + NbL\inp{\frac{(N + L)^L}{L!}}^{2}\inp{e(1 +
      \frac{N}{L(L+1)})}^{L(L+1)}}$$ 
      which
  computes  $\pi_w$ for  the Markov chain
  $\M_w$ described above.  For fixed $L$, $T= O(N^{O(L^2)})$. 
\end{theorem}

\noindent
The proof of the above theorem appears in Section
\ref{sec:proof-theor-refthm:p}, and is based on the projected RSM
process discussed in Section \ref{sec:prel-defin}.
The above theorem immediately gives a grid-search based algorithm that 
given a grid resolution $\delta$ and $\epsilon > 0$, outputs a
approximation $\mu_0$ to the error threshold in time
$T\cdot\nfrac{1}{2\delta}$ such that $\mu_0 \geq \mu_c^h(\epsilon, N)$
and $d^h(\vec{D}_{\infty, \mu_0 - \delta}, \vec{U}) > \epsilon$.   We now consider Markov Chain Monte Carlo based grid-search methods.  

\subsubsection{Markov Chain Monte Carlo Methods}
\label{sec:algor-estim-error-2}
The general strategy for Monte Carlo based grid search methods for
determining error thresholds is described in the algorithm \textsc{ErrorThreshold}
in Figure \ref{fig:algorithm} in the Appendix.  We will denote the
mixing time $\tau_{\rm mix}(\epsilon)$ for parameters $L, N, A$ and
$\mu$ as $\tau(L, N, A, \mu, \epsilon)$.  We consider 
the projected chain $\M_w$ described above which contains enough
information to compute the average Hamming weight, and whose state can be
maintained as a tuple in $\{0,1,\ldots,N\}^{L+1}$.

\begin{theorem}
\label{thm:algor-estim-error}
Let $A$ be class invariant, and consider the error threshold
$\mu_c^h(\epsilon, N)$.  Suppose the algorithm \textsc{ErrorThreshold} is run with
input grid resolution $\delta$, accuracy parameter $\delta_1$, and
error probability $\delta_2$. Let $T$ be the maximum over $k$ of the
quantity $\tau(L,N,A,k\delta,\delta_1/(2L))$ where $k \leq
1/(2\delta)$ is a positive integer.  The algorithm {\sc
  ErrorThreshold} runs in time $T \cdot s \cdot
\tilde{O}(\ceil{1/(2\delta)}NL\max_\sigma a_\sigma)$, where 
{$$
s = \ceil{\frac{8L^4}{\delta_1^2}\log\inp{2\ceil{1/(2\delta)}(L+1)/\delta_2}},
$$}
and with
probability at least $1-\delta_2$, produces an output $\mu_0$ satisfying $\mu_0 \geq \mu_c^h(\epsilon + \delta_1/2)$ and \\
$d^h(\vec{D}_{\infty, \mu_0 - \delta}, \vec{U}) \geq
\epsilon -\delta_1/2$.
\end{theorem}
\noindent
The proof of the above theorem appears in Section
\ref{sec:proof-theor-refthm}, where we also point out some technical
subtleties about the definition of error thresholds.

\paragraph{\bf Acknowledgments.} We thank Rajesh Balagam for helping
us with several simulations based on which this theoretical study was
initiated.  This work was initiated while Piyush Srivastava was an
intern at Microsoft Reseach, Bangalore.

\bibliography{bio}

\newcommand{\etalchar}[1]{$^{#1}$}
\begin{thebibliography}{GPLL{\etalchar{+}}05}

\bibitem[AB05]{althaus05}
Christian~L. Althaus and Sebastian Bonhoeffer.
\newblock {Stochastic Interplay between Mutation and Recombination during the
  Acquisition of Drug Resistance Mutations in Human Immunodeficiency Virus Type
  1}.
\newblock {\em Journal of Virology}, 79(21):13572--13578, 2005.

\bibitem[ADL04]{ADL04}
Jon~P. Anderson, Richard Daifuku, and Lawrence~A. Loeb.
\newblock Viral error catastrophe by mutagenic nucleosides.
\newblock {\em Annual Review of Microbiology}, 58(1):183--205, 2004.

\bibitem[AF98]{alves98}
D.~Alves and J.~F. Fontanari.
\newblock Error threshold in finite populations.
\newblock {\em Physical Review E}, 57(6):7008 -- 7013, June 1998.

\bibitem[BCP{\etalchar{+}}04]{Bonhoeffer04}
Sebastian Bonhoeffer, Colombe Chappey, Neil~T. Parkin, Jeanette~M. Whitcomb,
  and Christos~J. Petropoulos.
\newblock Evidence for positive epistasis in {HIV}-1.
\newblock {\em Science}, 306(5701):1547--1550, 2004.

\bibitem[BD97]{bubley97}
R.~Bubley and M.~E. Dyer.
\newblock Path coupling: a technique for proving rapid mixing in markov chains.
\newblock In {\em Proceedings of the 38th IEEE Symposium on the Foundations of
  Computer Science(FOCS)}, pages 223 -- 231, 1997.

\bibitem[BK98]{BK98}
D.~Bonnaz and A.~J. Koch.
\newblock Stochastic model of evolving populations.
\newblock {\em Journal of Physics A: Mathematical and General}, 31(2):417,
  1998.

\bibitem[BKP{\etalchar{+}}11]{batorsky11}
Rebecca Batorsky, Mary~F. Kearney, Sarah~E. Palmer, Frank Maldarelli, Igor~M.
  Rouzine, and John~M. Coffin.
\newblock {Estimate of effective recombination rate and average selection
  coefficient for HIV in chronic infection}.
\newblock {\em Proceedings of the National Academy of Sciences},
  108(14):5661--5666, 2011.

\bibitem[BS93]{BS93}
Sebastian Bonhoeffer and Peter~F. Stadler.
\newblock Error thresholds on correlated fitness landscapes.
\newblock {\em Journal of Theoretical Biology}, 164(3):359 -- 372, 1993.

\bibitem[BSSD11]{Balagam11}
Rajesh Balagam, Vasantika Singh, Aparna~Raju Sagi, and Narendra~M. Dixit.
\newblock Taking multiple infections of cells and recombination into account
  leads to small within-host effective-population-size estimates of {HIV}-1.
\newblock {\em PLoS ONE}, 6(1):e14531, 01 2011.

\bibitem[CCA01]{Crotty01}
Shane Crotty, Craig~E. Cameron, and Raul Andino.
\newblock {RNA} virus error catastrophe: Direct molecular test by using
  ribavirin.
\newblock {\em Proceedings of the National Academy of Sciences},
  98(12):6895--6900, 2001.

\bibitem[CF99]{CF99}
PRA Campos and JF~Fontanari.
\newblock Finite-size scaling of the error threshold transition in finite
  populations.
\newblock {\em J. Phys A}, 32:L1--L7, 1999.

\bibitem[DP09]{dubhashi09}
Devdatt~P. Dubhashi and Alessandro Panconesi.
\newblock {\em Concentration of Measure for the Analysis of Randomized
  Algorithms}.
\newblock Cambridge University Press, 2009.

\bibitem[Eig71]{eigen71}
M.~Eigen.
\newblock Selforganization of matter and the evolution of biological
  macromolecules.
\newblock {\em Die Naturwissenschaften}, 58:456--523, 1971.

\bibitem[EMS89]{EMS89}
M.~Eigen, J.~McCaskill, and P.~Schuster.
\newblock The molecular quasi-species.
\newblock {\em Adv. Chem. Phys.}, 75:149--263, 1989.

\bibitem[GD10]{GD10}
Saikrishna Gadhamsetty and Narendra~M. Dixit.
\newblock Estimating frequencies of minority nevirapine-resistant strains in
  chronically {HIV}-1-infected individuals naive to nevirapine by using
  stochastic simulations and a mathematical model.
\newblock {\em J. Virol.}, 84(19):10230--10240, 2010.

\bibitem[GKB{\etalchar{+}}05]{Gonzalez05}
Enrique Gonzalez, Hemant Kulkarni, Hector Bolivar, Andrea Mangano, Racquel
  Sanchez, Gabriel Catano, Robert~J. Nibbs, Barry~I. Freedman, Marlon~P.
  Quinones, Michael~J. Bamshad, Krishna~K. Murthy, Brad~H. Rovin, William
  Bradley, Robert~A. Clark, Stephanie~A. Anderson, Robert~J. O'Connell,
  Brian~K. Agan, Seema~S. Ahuja, Rosa Bologna, Luisa Sen, Matthew~J. Dolan, and
  Sunil~K. Ahuja.
\newblock The influence of {CCL3L1} gene-containing segmental duplications on
  {HIV}-1/{AIDS} susceptibility.
\newblock {\em Science}, 307(5714):1434--1440, 2005.

\bibitem[GPLL{\etalchar{+}}05]{GP05}
Ana Grande-P\'erez, Ester L\'azaro, Pedro Lowenstein, Esteban Domingo, and
  Susanna~C. Manrubia.
\newblock Suppression of viral infectivity through lethal defection.
\newblock {\em Proceedings of the National Academy of Sciences of the United
  States of America}, 102(12):4448--4452, 2005.

\bibitem[HC06]{citeulike:1993435}
Daniel~L. Hartl and Andrew~G. Clark.
\newblock {\em {Principles of Population Genetics, Fourth Edition}}.
\newblock Sinauer Associates, Inc., 4th edition, December 2006.

\bibitem[HMC{\etalchar{+}}11]{Hinkley}
Trevor Hinkley, Joao Martins, Colombe Chappey, Mojgan Haddad, Eric Stawiski,
  Jeannette~M Whitcomb, Christos~J Petropoulos, and Sebastian Bonhoeffer.
\newblock A systems analysis of mutational effects in {HIV}-1 protease and
  reverse transcriptase.
\newblock {\em Nature Genetics}, 43:487--489, 2011.

\bibitem[JS88]{jerrum88}
Mark Jerrum and Alistair Sinclair.
\newblock Conductance and the rapid mixing property for markov chains: the
  approximation of the permanent resolved.
\newblock In {\em Proceedings of the 20th Annual ACM Symposium on Theory of
  Computing (STOC 1988)}, pages 235--243, 1988.

\bibitem[KAB06]{Kouyos07}
Roger~D. Kouyos, Christian~L. Althaus, and Sebastian Bonhoeffer.
\newblock Stochastic or deterministic: what is the effective population size of
  {HIV}-1?
\newblock {\em Trends in Microbiology}, 14(12):507 -- 511, 2006.

\bibitem[LA10]{LA10}
Adam~S. Lauring and Raul Andino.
\newblock Quasispecies theory and the behavior of {RNA} viruses.
\newblock {\em PLoS Pathog}, 6(7):e1001005, 07 2010.

\bibitem[MHH{\etalchar{+}}11]{10.1371/journal.pone.0015135}
James~I. Mullins, Laura Heath, James~P. Hughes, Jessica Kicha, Sheila Styrchak,
  Kim~G. Wong, Ushnal Rao, Alexis Hansen, Kevin~S. Harris, Jean-Pierre Laurent,
  Deyu Li, Jeffrey~H. Simpson, John~M. Essigmann, Lawrence~A. Loeb, and Jeffrey
  Parkins.
\newblock Mutation of {HIV}-1 genomes in a clinical population treated with the
  mutagenic nucleoside kp1461.
\newblock {\em PLoS ONE}, 6(1):e15135, 01 2011.

\bibitem[Mus11]{musso11:}
Fabio Musso.
\newblock A stochastic version of {Eigen}'s model.
\newblock {\em Bulletin of Mathematical Biology}, 73:151 -- 180, 2011.

\bibitem[NBS{\etalchar{+}}98]{Nijhuis98}
Monique Nijhuis, Charles A.~B. Boucher, Pauline Schipper, Thomas Leitner, Rob
  Schuurman, and Jan Albert.
\newblock Stochastic processes strongly influence {HIV}-1 evolution during
  suboptimal protease-inhibitor therapy.
\newblock {\em Proceedings of the National Academy of Sciences},
  95(24):14441--14446, 1998.

\bibitem[NS89]{NS89}
M.~Nowak and P.~Schuster.
\newblock Error thresholds of replication in finite populations-mutation
  frequencies and the onset of {Muller}'s ratchet.
\newblock {\em J. Theor Biol}, 137:375--395, 1989.

\bibitem[PMnD10]{PMD10}
Jeong-Man Park, Enrique Mu\~noz, and Michael~W. Deem.
\newblock Quasispecies theory for finite populations.
\newblock {\em Phys. Rev. E}, 81(1):011902, Jan 2010.

\bibitem[SH06]{saakian06}
David~B. Saakian and Chin-Kun Hu.
\newblock Exact solutions of the {Eigen} model with general fitness functions
  and degradation rates.
\newblock {\em PNAS}, 103(13):4935 -- 4939, 2006.

\bibitem[SRA08]{Saakian08}
David~B. Saakian, Olga Rozanova, and Andrei Akmetzhanov.
\newblock Dynamics of the {Eigen} and the {Crow-Kimura} models for molecular
  evolution.
\newblock {\em Phys. Rev. E}, 78(4):041908, Oct 2008.

\bibitem[SS82]{SS82}
J.~Swetina and P.~Schuster.
\newblock Self-replication with errors: a model for polynucleotide replication.
\newblock {\em Biophys. Chem.}, 16:329--345, 1982.

\bibitem[TBVD12]{SimulationsPaper12}
Kushal Tripathi, Rajesh Balagam, Nisheeth~K. Vishnoi, and Narendra~M. Dixit.
\newblock Stochastic simulations suggest that {HIV}-1 survives close to its
  error threshold.
\newblock {\em Submitted}, 2012.

\bibitem[TH07]{TH07}
Nobuto Takeuchi and Paulien Hogeweg.
\newblock Error-threshold exists in fitness landscapes with lethal mutants.
\newblock {\em BMC Evolutionary Biology}, 7(1):15, 2007.
\newblock A response to Claus Wilke: Quasispecies theory in the context of
  population genetics, BMC Evol Biol 2005, 5:44.

\bibitem[vNCM99]{nimwegen99}
Eric van Nimwegen, Japen~P. Crutchfield, and Melanie Mitchell.
\newblock Statistical dynamics of the royal road genetic algorithm.
\newblock {\em Theoretical Computer Science}, 229:41 -- 102, 1999.

\bibitem[WH96]{WH96}
G.~Woodcock and PG~Higgs.
\newblock Population evolution on a multiplicative single-peak fitness
  landscape.
\newblock {\em J. Theor. Biol.}, 179:61--73, 1996.

\bibitem[Wie97]{Thomas97}
Thomas Wiehe.
\newblock Model dependency of error thresholds: the role of fitness functions
  and contrasts between the finite and infinite sites models.
\newblock {\em Genetics Research}, 69(02):127--136, 1997.

\bibitem[Wil05]{WilkeCommentary}
Claus Wilke.
\newblock Quasispecies theory in the context of population genetics.
\newblock {\em BMC Evolutionary Biology}, 5(1):44, 2005.

\bibitem[WK93]{WK93}
G.~P. Wagner and P.~Krall.
\newblock What is the difference between models of error thresholds and
  {Muller}'s ratchet?
\newblock {\em J. Math. Biol.}, 32:33--44, 1993.

\end{thebibliography}
\bibliographystyle{alpha}
\newpage

\appendix

\section{Starting State and Transition Matrix of the RSM Markov
  Chain}
\label{sec:start-state-trans}
As stated before, the RSM Markov chain starts with the ``fittest''
possible population with all the weight concentrated on the master
sequence, so that $\nt{0}{M} = N$ and $\nt{0}{\sigma} = 0$ for all
$\sigma \neq M$.  We now proceed to set up some notation for writing
out the transition matrix $\mathcal{M}$.

\begin{definition}
  \textbf{(Multivariate Geometric distribution).}  Let
  $g(\sigma)$ denote the number of genomes of type $\sigma$ in an urn.
  Consider the process of choosing, without replacement, $N$ genomes
  from this urn. Then $\phyp{g}{f}{N}$ denotes the probability of
  obtaining $f(\sigma)$ genomes of type $\sigma$ for each $\sigma$.
  We have,
\begin{equation}
  \label{eq:phyp}
  \phyp{g}{f}{N} \defeq
  \frac{\prod_{\sigma\in\inb{0,1}^L}\binom{g(i)}{f(i)}}{\binom{\inner{g}{\ones}}{N}}
\end{equation}
\end{definition}

\begin{definition}
  \textbf{(Multivariate Binomial Distribution).} Let
  $f(\sigma)$ denote the number of genomes of type $\sigma$.  Consider
  a stochastic process in which each genome of type $\sigma$
  independently mutates into a genome $\tau$ (possibly equal to
  $\sigma$) with probability $Q(\sigma,\tau)$.  We denote by
  $\pbin{f}{D}{Q}$ the probability that $D(\sigma, \tau)$ genomes of
  type $\sigma$ mutate to type $\tau$ under this process. We have
  \begin{equation*}
    \pbin{f}{D}{Q} \defeq
    \prod_{\sigma\in\inb{0,1}^L}\binom{f(\sigma)}{\inb{D(\sigma,\tau)|\tau\in\inb{0,1}^L}}\prod_{\tau\in\inb{0,1}^L}Q(\sigma,\tau)^{D(\sigma,\tau)}
  \end{equation*}
\end{definition}
\noindent
We can now write the entries of $\mathcal{M}$.  For $f, g :
\inb{0,1}^L \longrightarrow N$ satisfying $\inner{f}{\ones} = \inner{g}{\ones}
= N$, we denote by $\mathcal{M}(f,g)$ the conditional probability of
obtaining $g$ starting from $f$ in one step of the RSM process.  Given
a function $f \inb{0,1}^L\longrightarrow N$, we denote by $Af$ the function
such that $Af(\sigma) = a_\sigma f(\sigma)$.  Then, we have

\begin{equation*}
  \mathcal{M}(f,g) = \sum_{h:
    \inner{h}{\ones}=N}\phyp{Af}{h}{N}\sum_{D: \ones D = g; D \ones = h}\pbin{h}{D}{Q},
\end{equation*}
where $Q$ and $A$ are as defined above.
\pagebreak

\section{Proofs Omitted from Section \ref{sec:form-stat-main}}
\label{sec:proofs-omitted-from-2}
\subsection{Proofs Omitted from Section \ref{sec:prel-defin}}
\label{sec:proofs-omitted-from-1}
\begin{proof}[Proof Sketch of Fact \ref{fct:ergodic-chain}.]
  When $\mu \in (0,1) $ and $a_{\sigma}>0$ for all $\sigma,$ it can be
  verified easily that this chain is irreducible and also has a
  non-zero self-loop probability at every point in the state space.
  Thus, the chain is \emph{ergodic} and hence by the Fundamental
  theorem of Markov chains, has a unique stationary distribution to
  which it converges as $t\rightarrow\infty$.
\end{proof}

We now give a proof of Lemma \ref{lem:proj-strong}.  

\begin{proof}[Proof of Lemma \ref{lem:proj-strong}]
We will show that under class invariance, we can project the RSM
Markov chain so that its state space consists of equivalence classes
indexed by functions $h: [0,L] \longrightarrow N$ with $\sum_{i=0}^L h(i) = N$,
such that a function $f$ in the state space of $\mathcal{M}$ is in the
equivalence class $[h]$ if and only if for every $i \in [0,L]$,
$$
\sum_{\inb{\sigma\in\inb{0,1}^L|\w\sigma = i}} f(\sigma) = h(i).
$$

We will find it convenient to consider the reproduction and selection
phases separately from the mutation phase, show that the projection
described above can be done for both of them, and then combine the two
results using the following general fact about projected Markov
chains, the proof of which we include for completeness.

\begin{fact}
\label{fct:compose-projections}
  Let $P$ and $R$ be the transition kernels of two Markov chains on
  the same state space $\Omega$, and let $S$ denote the composition
  $PR$ of the two chains.  Suppose that there is a partition of
  $\Omega$ into equivalence classes $\Omega'$, such that for any $f
  \equiv f'$, and any equivalence class $[g]$, we have
$$
P(f, [g]) = P(f',[g])\text{ and }R(f, [g]) = R(f',[g]).
$$
Then, we also have $S(f, [g]) = S(f'([g])$, for all $f,f'$ and $g$ as
described above.
\end{fact}

\begin{proof}
  The proof is by direct computation.  We have,
  \begin{eqnarray}
    S(f, [g]) &=& \sum_{q'\in\Omega}P(f,q')R(q',[g])\nonumber\\
    &=& \sum_{[q]\in\Omega'}\sum_{q'\in[q]}P(f,q')R(q',[g])\nonumber\\
    &=& \sum_{[q]\in \Omega'}R(q,[g])\sum_{q'\in[q]}P(f,q')\nonumber\\
    &=& \sum_{[q]\in\Omega}R(q,[g])P(f,[q])\label{eq:4}\\
    &=& \sum_{[q]\in\Omega}R(q,[g])P(f',[q])\label{eq:5}
  \end{eqnarray}
Just as in the derivation of equation (\ref{eq:4}) above, we get 
$S(f', [g]) = \sum_{[q]\in\Omega'}R(q,[g])P(f',[q])$, and hence, by
equation (\ref{eq:5}), we have $S(f',[g]) = S(f,[g])$, as claimed.
\end{proof}

In order to use the last fact, we now decompose the matrix of the RSM
process into the following two Markov chains on $\Omega$:
\begin{enumerate}
\item \textbf{The Reproduce-Select Chain.} We denote the transition
  matrix of this chain as $P$, such that $P(f,g)$ is the probability
  of obtaining the state $g$ starting from state $f$ after the
  reproduction and selection phases.  Notice that 
$$
P(f, g) = \phyp{Af}{g}{N}.
$$
Assume that $A$ is class invariant and let $A(i)$ denote the
reproduction rate for a genome of Hamming weight $i$.  For an
equivalence class $[h]$ of $\Omega$ as defined above, we consider the
probability $P(f,[h])$, with $f \in [h']$.  By the definition of
$\phyp{}{}{}$, this is the probability of drawing $h(i)$ genomes of
Hamming weight $i$ for $0\leq i \leq L$, when $N$ genomes are drawn
without replacement from a bag containing $A(i)\sum_{\sigma:\w\sigma =
  i}f(\sigma) = A(i)h'(i)$ genomes of weight $i$.  By definition, this
probability depends only on $h$ and the equivalence class $h'$ of
$f$, and hence $P(f, [h]) = P(g, [h])$ when $A$ is class invariant and
$f\equiv g$.

\item \textbf{The Mutation Chain.} We will directly write down the
  entries $R(f, [h'])$ for the probability of obtaining a state in the
  equivalence class $[h']$ starting from a state $f$ in the
  equivalence class $[h]$.  We will show now that we can write $R(f,
  [h'])$ in terms only of $h$ and $h'$, and hence $R(f, [h']) = R(g,
  [h'])$ for $f\equiv g$.  Denote by $\mathcal{Q}_{ij}$ the
  probability that a string $\sigma$ of Hamming weight $i$ transforms
  into some string of Hamming weight $j$ in the mutation step, and
  notice that this probability is well defined because of the
  definition of the mutation transition matrix $Q$.  Since $f \in
  [h]$, there are $h(i)$ strings of Hamming weight $i$ initially, for
  $0\leq i \leq L$.  Denote by $d_{ij}$ the number of strings of
  weight $i$ which transform into strings of weight $j$ in the
  mutation step.  Then, we have
  \begin{equation}
    R(f,[h']) = \sum_{\substack{d: \sum_{j}d_{ij} = h(i)\\\sum_{i}d_{ij} =
        h'(j)}}\prod_{0\leq i\leq L}\binom{h(i)}{\inb{d_{ij}|0\leq j \leq L}}
    \prod_{0 \leq j \leq L}\mathcal{Q}_{ij}^{d_{ij}}.\label{eq:6}
\end{equation}
Since $R(f, [h'])$ depends only upon $h$ and $h'$, we get that $R(f,
[h']) = R(g, [h'])$ for $f\equiv g$. 
\end{enumerate}

\noindent Combining the above two discussions and using Fact
\ref{fct:compose-projections}, we see that when $A$ is class
invariant, the transition matrix $\mathcal{M}$ of the RSM process
satisfies $\mathcal{M}(f, [h']) = \mathcal{M}(g, [h'])$ whenever
$f\equiv g$.  This completes the proof of  Lemma \ref{lem:proj-strong}.
\end{proof}

\subsubsection{Relationships between Error Thresholds}
\label{sec:relat-betw-error}

We first define error thresholds according to various dispersal
measures.  
\begin{definition}[Error Thresholds] 
\label{def:error-thresholds-finite}
Let $\eps \geq 0$, and $\vec{U}$ be the uniform distribution over the set of
genomes. 
\begin{enumerate}
\item $\mu_{c}^{\rm ex,1}(\eps,N) \defeq \min \left\{ \mu \in (0,1):
    \left\|\E{\vec{D}_{\infty,\mu}} - \vec{U} \right\|_{1} \leq \eps
  \right\}$.
\item $\mu_{c}^{\rm h}(\eps,N) \defeq  \min \left\{ \mu \in (0,1):
    \left| \sum_{\sigma}\w{\sigma}\E{D_{\infty,\mu}^{\sigma}} -
      2^{-L}\sum_{\sigma}{\w{\sigma}} \right| \leq \eps \right\}$.
\item $\mu_{c}^{\rm sh}(\eps,N) \defeq \min \left\{ \mu \in (0,1):
    \left| \mathbf{H}(\E{\vec{D}_{\infty,\mu}}) - \mathbf{H}(\vec{U})
    \right| \leq \eps \right\}$.  Here $\mathbf{H}$ denotes the
  Shannon entropy, using the base $e$. 
\end{enumerate}
\end{definition}
\begin{definition}[Error Threshold for the Quasispecies Model]
  Let $\mu \in (0,1)$ and let $\vec{v}_{\mu}$ denote the the
  stationary expected fraction vector with $\ell_{1}$ norm $1.$ The
  error thresholds are defined as follows.  
\begin{enumerate}
\item  $ \mu_{c}^{\rm ex,1}(\eps) \defeq \min \{ \mu \in (0,1): \|\vec{v}_{\mu}  - \vec{U} \|_{1}\leq \eps \}.$ 
\item  $ \mu_{c}^{\rm h}(\eps) \defeq \min \{ \mu \in (0,1): |\sum_{\sigma}v_{\mu}(\sigma) w_{H}(\sigma)  - {2^{-L}}\sum_{\sigma}w_{H}(\sigma)  |\leq \eps \}.$ 
\item $ \mu_{c}^{\rm sh}(\eps) \defeq \min \{ \mu \in (0,1):|
  \mathbf{H}(\vec{v}_{\mu}) - \mathbf{H}(\vec{U}) | \leq \eps \},$ where $\vec{H}$ denotes the Shannon entropy, using the base $e$. 
\end{enumerate}
 \end{definition}
\noindent Our results are mostly stated in terms of the error threshold $\mu_c^h$.
However, we now describe how the different definitions above are related
to each other.
The following inequalities relate the different distance measures
that we have considered.  The first of these follows from the
definition of the $\ell_1$ norm, while the second is the well known
Pinsker's inequality relating the $\ell_1$ norm to the entropy.
\begin{eqnarray}
  \abs{\E{\sum_{\sigma}\w{\sigma}D_{\infty,\mu}^{\sigma}} - \E{\sigma
      \leftarrow U}{\sum_{\sigma}\frac{\w{\sigma}} {|\Omega|}}} &\leq&
  L\norm[1]{\E{\vec{D}_{\infty,\mu}}-\vec{U}}\label{eq:1}\\
  \norm[1]{\E{\vec{D}_{\infty,\mu}}-\vec{U}} &\leq& \sqrt{2\left| \mathbf{H}(\E{D_{\infty,\mu}})  - \mathbf{H}(\vec{U}) \right|}\label{eq:2}
\end{eqnarray}

\noindent This gives us the following relationship between the error-thresholds:
\begin{eqnarray}
\mu_c^{h}(\epsilon, N) &\leq& \mu_c^{\textrm{ex},1}(\epsilon/L, N)\nonumber\\
\mu_c^{\textrm{ex,1}}(\epsilon, N) &\leq & \mu_c^{\rm sh}(\epsilon^2/2,N)\label{eq:9}
\end{eqnarray}
\noindent Using the fact that the distributions involved are defined
over a state space of size $2^L$, we can show the following weak converse
to inequality (\ref{eq:2}):
\begin{equation*}
  \left| \mathbf{H}(\E{\vec{D_{\infty,\mu}}})  - \mathbf{H}(\vec{U}) \right| \leq 2^L  \norm[1]{\E{\vec{D}_{\infty,\mu}}-\vec{U}}^2.
\end{equation*}
\noindent This gives us a further relationship between the error thresholds:
\begin{equation*}
  \mu_c^{\textrm{ex,1}}(\epsilon, N) \geq \mu_c^{\rm sh}(2^L\epsilon^2,N)
\end{equation*}
However, we notice that one cannot in general close the loop in
inequalities (\ref{eq:1}) and (\ref{eq:2}) (and hence in inequalities
(\ref{eq:9})) above.  To see this, consider for example the following
two distributions $P$ and $Q$ for $L > 1.$
\begin{enumerate}
\item $P$: puts total weight $1-\epsilon$ on weight $1$ strings and weight $\epsilon$ on the string $\vec{0}$, so that the average Hamming weight is $1-\epsilon.$
\item $Q$: puts total weight $(1-\epsilon)/L$ on weight $L$ strings,
  and weight $1-(1-\epsilon)/L$ the string $\vec{0}$, so that the
  average Hamming weight is still $1-\epsilon$.
\end{enumerate}
The average Hamming weight in both cases is $1-\epsilon,$ so that in
that metric, the distance between $P$ and $Q$ is zero.  However, the
total variation distance between $P$ and $Q$ is at least $1-\epsilon.$

\subsection{Proof of  Theorem \ref{thm:convg-claim}}
\label{sec:proof-theor-refthm:c}
In the rest of this section, we will use the following concentration
inequalities about the multivariate hypergeometric distribution:
\begin{fact}
\label{fct:conc-hyp}
Consider the hypergeometric distribution $\phyp{g}{f}{N}$ defined in
equation (\ref{eq:phyp}) above. Let $D^\sigma$ be the random variable
denoting the fraction of genomes of type $\sigma$ which are drawn in
the process starting with $g(\tau)$ genomes of each type $\tau$.  We
then have:
  \begin{enumerate}
  \item $\E{D^\sigma} = \frac{g(\sigma)}{\inner{g}{\ones}}$.
  \item The following concentration inequality holds for $\epsilon
    \geq 0$:
    \begin{equation*}
      \Pr{|D^\sigma - \E{D^\sigma}| > \epsilon} \leq
      2\exp\inp{-\epsilon^2N}.
    \end{equation*}
  \end{enumerate}
\end{fact}

Similarly for the multivariate binomial distribution, we have the
following:
\begin{fact}
\label{fct:conc-bino}
Consider $N$ genomes with $f(\sigma)$ genomes of each type $\sigma$.
Let $D^\sigma$ be the random variable denoting the fraction of genomes
of type $\sigma$ after a mutation step. Then
  \begin{enumerate}
  \item $\E{D^\sigma} = \frac{1}{N}\sum_{\sigma}f(\tau)Q_{\tau\sigma}$.
  \item The following concentration inequality holds for $\epsilon
    \geq 0$:
    \begin{equation*}
      \Pr{|D^\sigma - \E{D^\sigma}| > \epsilon} \leq
      2\exp\inp{-2\epsilon^2N}
    \end{equation*}
  \end{enumerate}
\end{fact}
Fact \ref{fct:conc-hyp} is a consequence of Azuma's inequality, and a
proof can be found in the book by Dubhashi and Panconesi
\citep{dubhashi09}.  Fact \ref{fct:conc-bino} is essentially a
restatement of the Chernoff-Hoeffding bound.  Combining the above
bounds, we can deduce the following concentration inequality for each
step of the RSM process:
\begin{lemma}
\label{lem:concentration}
  Consider a state $\vec{N_t}$ of the RSM process.  We then have
  \begin{enumerate}
  \item $\E{\dt{t+1}{\sigma}|\vec{N_t}} =
    \frac{(\vec{D_t}AQ)_\sigma}{\inner{\vec{D_t}}{A}} =
    r^\sigma(\vec{D_t})$, with $r^\sigma$ as defined in equation
    (\ref{eq:r-def}).
  \item Let $\epsilon_1$ and $\epsilon_2$ be arbitrary positive
    constants. Then with probability (conditional on $\vec{N_t}$) at
    least $1 - 2^{2L+1}(\exp\inp{-\epsilon_1^2N} +
    \exp\inp{-2\epsilon_2^2N})$, we have $|\dt{t+1}{\sigma} -
    \E{\dt{t+1}{\sigma}|\vec{N_t}}| \leq (\epsilon_1 + \epsilon_2)$
    for every $\sigma$.  In particular, choosing $\epsilon_1 =
    \epsilon_2 = \epsilon/2$, we get that with probability
    (conditional on $\vec{N}_t$) at least $1 -
    2^{2L+2}\exp\inp{-\epsilon^2N/4}$, we have $|\dt{t+1}{\sigma} -
    \E{\dt{t+1}{\sigma}|\vec{N_t}}| \leq \epsilon$ for every $\sigma$.
  \end{enumerate}
\end{lemma}
\begin{proof}
  For ease of notation, let $g_\sigma = a_\sigma\nt{t}{\sigma}$.  Let
  $I^\sigma$ be the random variable denoting the fraction of genomes
  of type $\sigma$ left after the selection step.  Thus, we have
$$
\E{I^\sigma|\vec{N}_t} = \frac{g_\sigma}{\inner{\vec{g}}{\ones}}\text{
  and } \E{\dt{t+1}{\sigma}|\vec{I}, \vec{N}_t} = \sum_{\tau}I^\tau
Q_{\tau,\sigma}.
$$

Using a union bound over all genome types with the concentration
inequality in Fact \ref{fct:conc-hyp}, we get that with probability
at least $1 - 2^{L+1}\exp\inp{-\epsilon_1^2N}$ conditioned on
$\vec{N}_t$, we have 
\begin{equation*}
\abs{I^\sigma \in \frac{g_\sigma}{\inner{\vec{g}}{\ones}}}\leq
\epsilon_1\text{, for all }\sigma.
\end{equation*}
We denote the above event by $\mathcal{E}$.  Now, we consider the
concentration of $\vec{D}_t$ conditioned on $\vec{I}$.  Using a union
bound over all genome types along with the concentration inequality in
Fact \ref{fct:conc-bino}, we get with probability at least $1 -
2^{L+1}\exp\inp{-2\epsilon_2^2N}$ conditioned on $\vec{I}$, we have 
\begin{equation*}
  \abs{\dt{t+1}{\sigma} - \sum_\tau{I^\tau}Q_{\tau\sigma}} \leq
  \epsilon_2\text{, for all }\sigma.
\end{equation*}
We denote the above event by $\mathcal{F}$.  With probability at least
$1 - 2^{2L+1}(\exp\inp{-\epsilon_1^2N} + \exp\inp{-2\epsilon_2^2N})$,
conditioned on $\vec{N}_t$, both $\mathcal{E}$ and $\mathcal{F}$
occur, and then we have, for all $\sigma$,
\begin{eqnarray}
  \abs{\dt{t+1}{\sigma} - \E{\dt{t+1}{\sigma}|\vec{N}_t}} & = &
  \abs{\sum_\tau Q_{\tau\sigma}(\dt{t+1}{\sigma} -
    \frac{g_\tau}{\inner{g}{\ones}})}\nonumber\\
& \leq & \abs{\dt{t+1}{\sigma} -\sum_{\tau}I^\tau Q_{\tau\sigma}} +  \sum_\tau Q_{\tau\sigma}\abs{I^\tau - \frac{g_\tau}{\inner{g}{\ones}}}\nonumber\\
&=&\epsilon_2 + \epsilon_1\sum_\tau Q_{\tau\sigma} = \epsilon_1 + \epsilon_2,\nonumber
\end{eqnarray}
which is what we sought to prove.
\end{proof}
Before proceeding, we need the following lemma:
\begin{lemma}
\label{lem:lipschitz}
Fix a fitness landscape $A$ with positive entries and a mutation
transition matrix $Q$ with $\mu < 1/2$.  The functions $r^\tau$
defined in equation (\ref{eq:r-def}) are Lipschitz with Lipschitz
constant
$$K = \frac{\max_\tau{a_\tau}}{\min_\tau{a_\tau}}\inp{(1-\mu)^L -
  \mu^L}$$ on the set of probability distributions over genomes.
\end{lemma}
\begin{proof}
  For a probability distribution $\vec{x}$ over genomes, we have
  \begin{eqnarray*}
    \abs{\pdiff{r^{\sigma'}(\vec{x})}{x_\sigma}} & = & \frac{a_\sigma}{\sum_{\tau}a_\tau x_\tau}
      \abs{Q_{\sigma\sigma'} - r^{\sigma'}(\vec{x})}\nonumber\\
      &\leq&\frac{\max_\tau{a_\tau}}{\min_\tau{a_\tau}}\inp{(1-\mu)^L - \mu^L},\nonumber
  \end{eqnarray*}
  where the last line follows by noticing that fact that for all
  $\vec{x}$ and all $\sigma'$, $\min_{\sigma,\tau}Q_{\sigma\tau} \leq
  r^{\sigma'}(\vec{x}) \leq \max_{\sigma,\tau}Q_{\sigma\tau}$, and
  $\min_{{\sigma\tau}}Q_{\sigma\tau} = \mu^L$, while
  $\max_{{\sigma\tau}}Q_{\sigma\tau} = (1-\mu)^L$.  Thus, by the mean
  value theorem, for any probability distributions $\vec{x}$ and
  $\vec{y}$ over genomes, we get
$$
\abs{r(\vec{x}) - r(\vec{y})} \leq K\norm[1]{\vec{x} - \vec{y}}. 
$$
\end{proof}
\begin{proof}[Proof (of Theorem \ref{thm:convg-claim})]
  Fix a time $t_0$.  In the rest of the proof, we drop the
  conditioning on the initial state being concentrated on the master
  sequence for ease of notation.  We will prove the following claim by
  induction for $0\leq t\leq t_0$:
  \begin{claim}
    \label{claim:induc}
    For every $\sigma \in \inb{0,1}^L$ and $0\leq t \leq t_0$, there
    exist $l_t^\sigma, u_t^\sigma$ and $p_t$ satisfying the conditions
    \begin{enumerate}
    \item $0\leq l_t^\sigma \leq u_t^\sigma \leq 1$, and $s_t
      \defeq \max_\sigma{u_t^\sigma - l_t^\sigma}$ and $p_t$ are $o_N(1)$.  Also,
      $\xt{t}{\sigma}$ lies in the interval $\insq{l_t^\sigma,
        u_t^\sigma}$.\label{claim:induc:item:1}
    \item With probability at least $1-p_t$, $\dt{t}{\sigma}$ lies in
      the interval $\insq{l_t^\sigma, u_t^\sigma}$ for all
      $\sigma$. \label{claim:induc:item:2}
    \end{enumerate}
  \end{claim}
  We first see how to finish the proof of Theorem \ref{thm:convg-claim}
  assuming Claim \ref{claim:induc}.  From item
  \ref{claim:induc:item:2} in Claim \ref{claim:induc}, and using
  $\xt{t}{\sigma} \in \insq{l_t^\sigma, u_t^\sigma}$ ,we get
\begin{equation}
  \label{eq:11}
  |\E{\dt{t_0}{\sigma}} - \xt{t_0}{\sigma}| \leq p_{t_0} + (1-p_{t_0})|u_{t_0}^\sigma -
  l_{t_0}^\sigma| \text{, for all $\sigma$.}
\end{equation}
Now item \ref{claim:induc:item:1} of the claim implies that the right
hand side of equation (\ref{eq:11}) goes to $0$ as
$N\rightarrow\infty$, which concludes the proof of Theorem
\ref{thm:convg-claim}, assuming Claim \ref{claim:induc}.
\end{proof}
We now proceed to prove Claim \ref{claim:induc}.
\begin{proof}[Proof of Claim \ref{claim:induc}]
  At $t$ = 0, we can set $l_t^\sigma = u_t^\sigma = \xt{t}{\sigma}$,
  and $p_t = 0$.  By the definition of the starting state, this
  satisfies the conditions claimed in the claim.  Now suppose that we
  have shown that with probability $1-p_t$, we have $\dt{t}{\sigma}
  \in \insq{l_t^\sigma, u_t^\sigma}$  for all $\sigma$.  We call the
  latter event $\mathcal{E}_t$.   Recall that 
$$
\E{\dt{t+1}{\sigma}|\vec{D}_t} = r^\sigma(\vec{D}_t),
$$
and define
$$
{l'}_{t+1}^\sigma = \min_{\inb{\vec{y} | y^\sigma\in\insq{l_t^\tau, u_t^\tau}}}r^\sigma(\vec{y});\;\;
  {u'}_{t+1}^\sigma = \max_{\inb{\vec{y} | y^\tau\in\insq{l_t^\tau, u_t^\tau}}}r^\sigma(\vec{y})
$$
Notice that $\xt{t+1}{\sigma} \in \insq{{l'}_{t+1}^\sigma,
  {u'}_{t+1}^\sigma}$. Also, because of the Lipschitz condition on the
function $r^\sigma$ shown in Lemma \ref{lem:lipschitz}, we have
${u'}_{t+1}^\sigma - {l'}_{t+1}^\sigma \leq 2^LKs_t$.  Now, we condition
on the event $\mathcal{E}_t$ defined above, and in this case, we have
$$\E{\dt{t+1}{\sigma}|\mathcal{E}_t} \in \insq{{l'}_{t+1}^\sigma,
    {u'}_{t+1}^\sigma}\text{, for all }\sigma.
$$
Choose $\epsilon(N) = N^{-1/3} = o_N(1)$, and set $l_{t+1}^\sigma =
{l'}_{t+1}^\sigma - \epsilon/2$, $u_{t+1}^\sigma = {u'}_{t+1}^\sigma +
\epsilon/2$.  Using the concentration result quoted in Lemma
\ref{lem:concentration}, we get that conditioned on $\mathcal{E}_t$,
with probability at least $1-p(N)$ where $p(N) =
\exp(-\Omega(N^{1/3}))= o_N(1)$,
$$
\dt{t+1}{\sigma} \in \insq{l_{t+1}^\sigma, u_{t+1}^\sigma}\text{, for
  all }\sigma.
$$
Now, we saw above that $\mathcal{E}_t$ occurs with probability at
least $1-p_t$. Hence, by a union bound, we get that with probability
at least $1-p_{t+1}$, where $p_{t+1} = p_{t} + p(N)$,
$$
\dt{t+1}{\sigma} \in \insq{l_{t+1}^\sigma, u_{t+1}^\sigma}\text{, for
  all }\sigma.
$$
This proves the induction hypothesis, except that we need to make sure
that $s_t, p_t$ are $o_N(1)$.  We first consider $s_t$. From above, we
have the following recurrence for $s_t$:
\begin{equation}
s_{t+1} \leq 2^LKs_t + \epsilon(N); \;\; s_0 = 0.\label{eq:7}
\end{equation}
This satisfies $s_t = O_N(\epsilon) = o_N(1)$ for all $t \leq t_0$, by
the choice of $\epsilon$.  Similarly, we have $p_t = tp(N) = o_N(1)$ for
$t \leq t_o$ by the choice of $p(N)$.  This proves Claim
\ref{claim:induc}.
\end{proof}

\subsection{Proof of Corollary \ref{coroll:conv-error-threshold}}
\label{sec:error-threshold-convergence}

We begin by noticing that for $0 < \mu < 1/2$, we can choose a time
$t_0$ such that for $t > t_0$, the state $\vec{m}_t$ of the
quasispecies model satisfies
\begin{equation}
\abs{d^h(\vec{m}_t, \vec{U}) - d^h(\vec{v}_\mu,
  \vec{U})} \leq \delta/2,\label{eq:14}
\end{equation}
where $\vec{v}$ is the unique stationary vector of the quasispecies
model.  Now fix $t > t_0$. Since the distance function $d^h$ is
continuous, Theorem \ref{thm:convg-claim} allows us to choose an
$N_\delta$ such that for $N > N_\delta$,
\begin{equation}
  \label{eq:15}
  \abs{d^h(\vec{m}_t, \vec{U}) - d^h(\E{\vec{D}_{t,\mu,N}},
  \vec{U})} \leq \delta/2.
\end{equation}
Combining equations (\ref{eq:14}) and (\ref{eq:15}), we get 
$$
\abs{d^h(\vec{v}, \vec{U}) - d^h(\E{\vec{D}_{t,\mu,N}},
  \vec{U})} \leq \delta.
$$
Thus, when $\mu < \mu_c^h(\epsilon)$, we have
\begin{displaymath}
d^h\inp{\E{\vec{D}_{t,\mu,N}}, \vec{U}} \geq
    \epsilon - \delta,\text{ when $\mu <
      \mu_c^{h}(\epsilon)$, and},
  \end{displaymath}
and when $\mu = \mu_c^h(\epsilon)$,
\begin{displaymath}
  d^h\inp{\E{\vec{D}_{t,\mu,N}}, \vec{U}} \leq
  \epsilon + \delta \text{ when $\mu =
    \mu_c^{h}(\epsilon)$.}
\end{displaymath}

\subsection{Proof of Theorem \ref{thm:mixing-main}}
\label{sec:proof-theor-refthm:m}

Before proving Theorem \ref{thm:mixing-main}, we first set up some
notation for the coupling argument.

\begin{definition} A coupling of two probability distributions
  $\mathcal{D}_{1}$ and $\mathcal{D}_{2}$ is a pair of random
  variables $(X, Y)$ defined on a single probability space such that
  the marginal distribution of $X$ is $\mathcal{D}_{1}$ and the
  marginal distribution of $Y$ is $\mathcal{D}_{2}.$ \end{definition}

\begin{definition} A coupling of Markov chains with transition matrix
$\mathcal{M}$ is defined to be a process $(X_{t},
Y_{t})_{t=0}^{{\infty}}$ with the property that both $(X_{t})$ and
$(Y_{t})$ are Markov chains with transition matrix $\mathcal{M},$
although the two chains may possibly have different starting
distributions.  
\end{definition}
 
Any coupling of Markov chains with transition matrix $\mathcal{M}$ can
be modified so that the two chains stay together at all times after
their first simultaneous visit to a single state: more precisely, such
that if $ X_{s} = Y_{s},$ then $X_{t} = Y_{t}$ for $t \geq s.$ In the
following, we only consider such couplings.  The following well known
facts are the basis of coupling based methods for proving mixing time
bounds.

\begin{theorem} Let $\{(X_{t}, Y_{t})\}$ be a coupling satisfying the
  definition above for which $X_{0} = \alpha$ and $Y_{0}= \beta.$ Let
  $\tau_{\rm couple}$ be the first time the chains meet: $ \tau_{\rm
    couple} \defeq \min \{t : X_{t} = Y_{t}\}.$ Then
$$ \|\mathcal{M}^{t}(\alpha, \cdot) - \mathcal{M}^{t}(\beta, \cdot)\|_{{TV}} \leq \Pr{\tau_{\rm couple} > t | X_{0}=\alpha,Y_{0}=\beta}.$$
 \end{theorem}

 \begin{lemma}[Coupling Lemma] Let $X,Y$ be random variables defined
   on a finite sample space $\Omega$ and let $\mathcal{C}$ be any
   coupling of $C$ and $Y.$ Then
$$ \min _{\mathcal{C}} \Psymb_{\mathcal{C}}[X \neq Y] = \|X-Y\|_{TV}.$$

\end{lemma}

\begin{definition} Let $d: \Omega \times \Omega \longrightarrow
\mathbb{R}_{\geq 0}$ be a distance metric on the state space
$\Omega$ of the two Markov chains
$\{X_{t}\}_{t}$ and $\{Y_{t}\}_{t}.$ Suppose $\mathcal{C}$ is a
coupling such that for every $t \geq 0,$
$$ \E{d(X_{t+1},Y_{t+1})}  \leq \theta \cdot \E{d(X_{t},Y_{t})} $$
for every starting distributions $X_{0},Y_{0},$ then we call
$\mathcal{C}$ a $(\theta,d)$ coupling. Note that this implies that
$$ \E{d(X_{t},Y_{t})} \leq \theta^{t} \cdot D,$$
where $D=\max_{\sigma,\tau \in \Omega} d(\sigma,\tau).$

\end{definition}

\noindent Fix a integer valued distance function $d$. Let let
$\{X_{t}\}_{t}$ be a realization of the Markov chain starting from
$X_{0}$ and $Y_{t}$ be another realization starting from the
stationary distribution $\pi$ of the Markov chain.  If $\mathcal{C}$
is a $(\theta,d)$ coupling, then it follows from the Coupling Lemma
that $$ \|X_{t}-\pi \|_{TV}\stackrel{\rm Coupling}{\leq} \Pr{X_{t}\neq
  Y_{t}} \stackrel{d \; {\rm integral}}{=} \Pr{d(X_{t},Y_{t})} \geq 1]
\stackrel{{\rm Markov}}{\leq} \E{d(X_{t},Y_{t})} \leq \theta^{t}\cdot
D.$$
This implies that the mixing time $\tau_{\rm mix}(\epsilon) =
O\inp{\frac{\log~\inp{D/\epsilon}}{\log~(1/\theta)}}$ when $\theta < 1$.

\subsubsection{A Coupling for the RSM Process}
The coupling $\mathcal{C}$ we will construct will have two independent
parts $\mathcal{C}=(\mathcal{C}_{S}, \mathcal{C}_{M}),$
$\mathcal{C}_{S}$ for the Reproduce-Select phase and $\mathcal{C}_{M}$
for the mutation part.  We first describe the somewhat simpler
mutation coupling.

\paragraph{\bf Mutation Coupling $\mathcal{C}_{M}$.} Let
$X_{t}=\{\sigma_{1},\ldots,\sigma_{N}\}$ and $Y_{t}=\{\nu_{1},\ldots,
\nu_{N}\}.$ Let $M$ denote an arbitrary permutation on $[N]$ such that
$M(i)$ denotes the image of $i \in [N]$.  We define the distance
between the states as
$$d_{\rm match}(X_{t},Y_{t} )\defeq \min_{M} \left\{  \sum_{i=1}^{N} \ham{\sigma_{i},\nu_{M(i)}} \right\}.$$
The mutation coupling follows the following algorithm, with $M$ set to
be the permutation which achieves the minimum in the above definition.
\begin{enumerate}
\item For $i=1,\ldots, N$
\begin{enumerate}
\item For $j=1,\ldots, L$
\begin{enumerate}
\item Choose independently and uniformly at random $r_{j}$ from
$[0,1].$
\item If $\sigma_{i}(j) = \nu_{M(i)} (j)$
\begin{enumerate}
\item Flip $\sigma_{i}(j)$ and $\nu_{M(i)} (j)$ if and only if $r_{j} \geq
1-\mu.$
\end{enumerate}
\item Else
\begin{enumerate}
\item Let $r_{j,1}\defeq r_{j}$ and $r_{j,2}\defeq 1-r_{j}.$
\item Flip $\sigma_{i}(j)$ if and only if $r_{j,1} \geq 1-\mu.$
\item Flip $\nu_{M(i)}(j)$ if and only if $r_{j,2} \geq 1-\mu.$
\end{enumerate}
\end{enumerate}
\end{enumerate}
\end{enumerate}
\begin{lemma} For $\mu \leq 1/2,$ $\mathcal{C}_{M}$ is a
$((1-2\mu),d_{\rm match})$ coupling.
\end{lemma}
\begin{proof} To prove that $\mathcal{C}_{M}$ is a valid coupling one just
  needs to note that if $r$ is distributed uniformly at random in
  $[0,1]$ then so is $1-r.$ Hence, for $i=1,\ldots,N$ and
  $j=1,\ldots,L,$ each bit $\sigma_{i}(j)$ (respectively,
  $\nu_{i}(j))$) flips with probability exactly $\mu.$ Further, these
  flips are independent by construction.

  To prove that $\mathcal{C}_{M}$ is a $((1-2\mu),d_{\rm match})$
  coupling, let $X_{t},Y_{t}$ be the states of the two Markov chains
  with distance $d \defeq d_{\rm match}(X_{t},Y_{t}).$ By definition
  of $d_{\rm match},$ there is some matching $M^{\star}$ which
  achieves $d.$ Without loss of generality assume that $M^{\star}$ is
  identity, i.e., $M^{\star}(i)=i,$ for all $1 \leq i \leq N.$ Hence,
  $\sum_{i=1}^{N} d_{H} (\sigma_{i},\nu_{i})=d.$ Let $X_{t+1} \defeq
  (\sigma_{1}^{t+1},\ldots,\sigma_{N}^{t+1})$ and $Y_{t+1}\defeq
  (\nu_{1}^{t+1},\ldots,\nu_{N}^{t+1})$ be the output of
  $\mathcal{C}_{M}$ on input $(\sigma_{1},\ldots,\sigma_{N})$ and
  $(\nu_{1},\ldots,\nu_{N})$ respectively.  We will calculate
  $\E{\sum_{i=1}^{N} d_{H} (\sigma_{i}^{t+1}, \nu_{i}^{t+1})}$ and
  show that it is exactly $(1-2\mu) \cdot d.$ Hence, $d_{\rm match}
  (X_{t+1},Y_{t+1}) \leq (1-2 \mu) \cdot d,$ as $d_{\rm match}$ is
  defined as the minimum over all possible matchings.  By linearity of
  expectation it is sufficient to show that for all $i=1,\ldots,N,$
$$ \E{d_{H}(\sigma_{i}^{t+1},\nu_{i}(t+1)}= (1-2\mu) \cdot d_{H}(\sigma_{i},\nu_{i}).$$
Again by linearity of expectation it is sufficient to show the
following:
$$ \Ex{r_{j}}{d_{H}(\sigma_{i}^{t+1}(j),\nu_{i}^{t+1}(j))} = (1-2\mu)\cdot d_{H}(\sigma_{i}(j),\nu_{i}(j)).$$
This follows from observing that if $\sigma_{i}(j)=\nu_{i}(j),$ then
$\Pr{\sigma_{i}^{t+1}(j)=\nu_{i}^{t+1}(j)} = 1,$ while if
$\sigma_{i}(j) \neq \nu_{i}(j),$ then, as $\mu \leq 1/2,$
$\Pr{\sigma_{i}^{t+1}(j)=\nu_{i}^{t+1}(j)} = 2\mu.$ Hence,
$\Pr{\sigma_{i}^{t+1}(j)\neq \nu_{i}^{t+1}(j)} = 1-2\mu.$ This
completes the proof.

\end{proof}

\paragraph{\bf Coupling $\mathcal{C}_S$ for the Selection Process.}
We again consider two states $X_t = \inb{\sigma_1, \sigma_2, \ldots,
  \sigma_N}$ and $Y_t = \inb{\nu_1, \nu_2,\ldots, \nu_N}$. Our
distance function is still $d_{\rm match}$ defined above.  We first
note the $\dm{\cdot, \cdot}$ is actually a metric.
\begin{lemma}
  $\dm{\cdot, \cdot}$ is a metric.
\end{lemma}
\begin{proof}
  By construction $\dm{X,Y} \geq 0$ with equality if and only if $X =
  Y$. Now consider states $X=\inb{\sigma_i}_{i=1}^N, Y =
  \inb{\tau_i}_{i=1}^N$ and $ Z= \inb{\nu_i}_{i=1}^N$ in the state
  space $\Omega$.  Let $\alpha$ and $\beta$ be permutations of $[N]$
  such that
$$
\dm{X,Y} = \sum_{i=1}^N\ham{\sigma_i, \tau_{\alpha(i)}}\;\;\text{and}\;\;
\dm{Y,Z} = \sum_{i=1}^N\ham{\tau_i, \nu_{\beta(i)}}\;\;
$$
Now we have
\begin{eqnarray}
  \dm{X,Z} & \leq &
  \sum_{i=1}^{N}\ham{\sigma_i,\nu_{\beta(\alpha(i))}}\nonumber\\
  &\leq & \sum_{i=1}^N\inp{\ham{\sigma_i, \tau_{\alpha(i)}} +
    \ham{\tau_{\alpha_i} + \nu_{\beta(\alpha(i))}}}\nonumber\\
  & = & \dm{X,Y} + \dm{Y,Z}.\nonumber
\end{eqnarray}
\end{proof}
We will use the following general path coupling result of Bubley and
Dyer~\citep{bubley97} to define the coupling $\mathcal{C}_S$.
\begin{theorem}[Path Coupling~\citep{bubley97}]
  Consider a Markov chain $M$ on state space $\Omega$ and a distance
  function $d$ on $\Omega$ such that $d'(x,x) = 0$ for all $x \in
  \Omega$.  Consider a connected undirected graph $G$ on $\Omega$ such
  that the length of each edge $\inb{x,y}$, if present in $G$, is
  $d(x,y)$, and let $d'$ be the shortest path metric on $G$.  Suppose
  there exists a coupling $\mathcal{C}$ for $M$ such that for some
  $\alpha < 1$, and all $X_t,Y_t \in \Omega$ which are adjacent in
  $G$,
$$
\Ex{\mathcal{C}}{d(X_{t+1},Y_{t+1})|X_t,Y_t} \leq \alpha d(X_t,Y_t).
$$
If every edge of $G$ is a shortest path under the metric $d'$
described above, then the coupling $\mathcal{C}$ can be extended to a
coupling $\mathcal{C'}$ such that 
$$
\Ex{\mathcal{C'}}{d'(X_{t+1},Y_{t+1}|X_t,Y_t} \leq \alpha d'(X_t,Y_t).
$$
for all $X_t, Y_t \in \Omega$.\label{thm:path-coupling}
\end{theorem}
We will first show now that the path metric resulting from an
application of the above theorem to $d_{\rm match}$ is $d_{\rm match}$
itself, since this is crucial for composing the $\mathcal{C}_S$
coupling with the coupling $\mathcal{C}_M$ described above.

\begin{lemma}
  Consider the state space $\Omega$ of the RSM Markov chain
  $\mathcal{M}$.  Let $G$ be the graph on $\Omega$ in which two
  vertices $X$ and $Y$ are adjacent if and only if $\dm{X,Y} = 1$.
  Then the path metric $d'$ constructed in Theorem
  \ref{thm:path-coupling} is identical with $d$, and each edge in $G$
  is a shortest path.
\end{lemma}

\begin{proof}
  For brevity we will denote $\dm{\cdot,\cdot}$ by $d$.  Notice that
  since each edge is of length $1$, it is also a shortest path by
  construction.  Since $d$ is a metric, we also have $d'(X,Y) \geq
  d(X,Y)$ for all $X,Y \in \Omega$.  We now proceed by induction to show
  that $d(X,Y) \geq d'(X,Y)$ for all $X,Y \in \Omega$.  Notice that
  when $d(X,Y) = 1$, this is true by definition of $d'$.  Now, suppose
  that $d(X, Y) \leq k - 1$ implies $d'(X,Y) \leq d(X,Y)$, and
  consider the case $d(X,Y) = k > 1$.  We claim that there exists a
  $Z$ such that $d(X, Z) \leq k-1$ and $d(Y,Z) \leq 1$.  The existence
  of such a $Z$ implies using the induction hypothesis that
\begin{eqnarray}
  d'(X,Y) &\leq& d'(X,Z) + d'(Z,Y)\nonumber\\
  &\leq& d(X,Z) + d(Z,Y) \leq k = d(X,Y).\nonumber 
\end{eqnarray}

It only remains to construct such a $Z$.  Let $X =
\inb{\sigma_i}_{i=1}^N$ and $Y = \inb{\nu_i}_{i=1}^N$.  Without loss
of generality, we may assume that $d(X,Y) = \sum_{i=1}^N\ham{\sigma_i,
  \nu_i}$.  Since $d(X, Y) = k > 1$, there exists a $j$ such that
$\ham{\sigma_j, \nu_j} \geq 1$. Let $s$ be a string obtained by
flipping a single bit of $\nu_j$ such that $\ham{\sigma_j, s} =
\ham{\sigma_j, \nu_j} - 1$.  Now let $Z = \inb{\tau_i}_{i=1}^N$, where
$\tau_i = \nu_i$ for $i\neq j$ and $\tau_j = s$.  By construction,
$d(Y,Z) \leq 1$ and $d(X, Z) \leq k-1$.

\end{proof}

\paragraph{\bf Claims about the Coupling.}
\label{sec:claims-about-coupl}
Suppose we find a coupling $C$ for the selection phase such that when
$\dm{X_t, Y_t} = 1$, then the intermediate states $I(X_t)$ and
$I(Y_t)$ satisfy $$\E{\dm{I(X_t), I(Y_t)}| X_t, Y_t} \leq \alpha,$$
then using Theorem \ref{thm:path-coupling} and the coupling for the
mutation phase described above, we get
$$\E{\dm{X_{t+1},Y_{t+1}|X_t,Y_t}} \leq \alpha(1-2\mu)\dm{X_t,Y_t}.$$
This will give us fast mixing as long as $\alpha(1-2\mu) < 1$.

We now describe such a coupling for the selection process, for general
$X_t$ and $Y_t$, which we will analyze only in the simple but
sufficient case when $\dm{X_t, Y_t} = 1$.  Suppose that $X_t$ and
$Y_t$ contain, respectively, $n_\sigma^x$ and $n_\sigma^y$ genomes of
type $\sigma$.  After reproduction, the number of genomes of type
$\sigma$ in the two chains is $a_\sigma n_\sigma^x$ and $a_\sigma
n_\sigma^y$ respectively.  Let the total number of genomes be $M_x =
\sum_\sigma a_\sigma n^x_\sigma$ and $M_y = \sum_\sigma a_\sigma
n^y_\sigma$ respectively.  Without loss of generality let us assume
that $M_x \geq M_y$, and set $M = M_x$.

We now construct a bag of $M$ balls as follows.  For each $\sigma$,
the bag has exactly $a_\sigma\min\inp{n^x_\sigma, n^y_\sigma}$ balls
with label $(\sigma, (x,y))$.  If $n^x_\sigma \geq n^y_\sigma$, then
the bag has exactly $a_\sigma(n^x_\sigma - n^y_\sigma)$ balls with
label $(\sigma, x)$, otherwise it has exactly $a_\sigma(n^y_\sigma -
n^x_\sigma)$ balls with label $(\sigma, y)$.  Thus the total number of
balls in the bag is $M$.

We now take a random permutation of the $M$ balls, and take the
intermediate state $I_X$ (respectively, $I_Y$) to be the multiset of
genomes given by the first $N$ balls carrying the label $(x,y)$ or $x$
(respectively, $(x,y)$ or $y$).  Notice that a ball carrying a label
$(x,y)$ can contribute a genome to both $I_X$ and $I_Y$.

\begin{claim}
  The above coupling is a valid Markovian coupling for the selection
  phase of the RSM chain.
\end{claim}

\begin{proof}
  Notice that sampling without replacement $a$ objects from a set of
  $b$ objects is equivalent to taking the first $a$ elements from a
  uniform random permutation of the $b$ objects.  Also note that given
  a subset $S$ of a set of $b$ objects, and a uniform random
  permutation $\alpha$ over the $b$ objects, the restriction of
  $\alpha$ to the elements of $S$ is a uniformly random permutation of
  the elements of $S$.  Now consider the set of $M$ labeled balls
  constructed above, and define $S_X$ (respectively, $S_Y$) to be the
  set of balls carrying a $(x,y)$ or $x$ (respectively, $x(x,y)$ or
  $y$) label.  By the observations above, see that the set $I_X$
  (respectively, $I_Y$) has the same distribution as if it was sampled
  without replacement from $S_X$ (respectively, $S_Y$).  This proves
  the claim.
\end{proof}
\begin{lemma}
Suppose $\dm{X_t, Y_t} = 1$.  Then under the above coupling, we have
$$
\E{\dm{I_X,I_Y}|X_t,Y_t} \leq \frac{1}{1-\frac{1}{N}}\frac{\max
  a_\sigma}{\min a_\sigma}L.
$$
\end{lemma}
\begin{proof}
  For brevity, we will denote $X_t$ and $Y_t$ as $X =
  \inb{\sigma_i}_{i=1}^N$ and $Y = \inb{\nu_i}_{i=1}^N$ and
  $\dm{\cdot,\cdot}$ by $d$.  Since $d(X,Y) = 1$, we can assume
  without loss of generality that $\sigma_i = \nu_i$ for $i > 1$, and
  $\sigma_1$ and $\nu_1$ differ in exactly one bit.   

  We now consider the coupling described above, and let $S$ be the set
  of balls with label $(x,y)$.  Notice that $|S| = \sum_{1<i\leq
    N}a_{\sigma_i}$. Notice that $|S| \geq (N-1)\min_\sigma a_\sigma$.  We
  assume without loss of generality that $a_{\sigma_1} > a_{\nu_1}$,
  so that the number of balls $M$ in the bag is $|S| + a_{\sigma_1}$.
  Consider a random permutation $\alpha$ of the $M$ balls, and let $I$
  be the (random) multiset of balls with label $(x,y)$ occurring in
  the first $N$ positions of $\alpha$.  Notice that $|I_X \cap I_Y|
  \geq |I|$.  We observe that if the intersection of $I_X$ and $I_Y$
  (seen as multisets) is of size at least $|I|$, then $\dm{I_X,I_Y}
  \leq L(N-|I|)$, hence we have
  \begin{equation}
    \label{eq:coupling:1}
    \E{\dm{I_X, I_Y}|X,Y} \leq L(N-\E{|I|~|X,Y}). 
  \end{equation}
  Now, we have 
  \begin{eqnarray}
    \E{|I|~|X,Y} & = & \frac{|S|}{|S| + a_{\sigma_1}}N\nonumber\\
    & \geq & N\inp{1 - \frac{a_{\sigma_1}}{|S|}} \geq N\inp{1 - \frac{a_{\sigma_1}}{(N-1)\min_\tau a_\tau}}.\nonumber
  \end{eqnarray}
Plugging this into equation (\ref{eq:coupling:1}), we get 
\begin{eqnarray}
  \E{\dm{I_X, I_Y}|X,Y} & \leq & \frac{LNa_{\sigma_1}}{(N-1)\min_\tau
    a_{\tau}}\nonumber\\
  & \leq & \frac{1}{1-\frac{1}{N}}\frac{\max_\tau a_\tau}{\min_\tau a_\tau}L,\nonumber
\end{eqnarray}
which establishes our claim.
\end{proof}

Using the above discussion, we get fast mixing under the condition
that 
$$
(1-2\mu) \frac{1}{1-\frac{1}{N}}\frac{\max_\sigma a_\sigma}{\min_\sigma
  a_\sigma}L < 1.
$$
Formally, we have
\begin{theorem}
  Fix a mutation rate $\mu < 1/2$, and a fitness landscape $A$.  Define 
$$
K(A, \mu) = (1-2\mu)\frac{1}{1-\frac{1}{N}}\frac{\max_\sigma
  a_\sigma}{\min_\tau a_\tau}L. 
$$
When $K(A, \mu) < 1$,  we have $\tau_{\rm mix}(\epsilon) =
O\inp{\frac{\log(NL/\epsilon)}{\log (1/K)}}$.  
\end{theorem}
We note here that the main difficulty in designing coupling arguments
for the RSM process is the distance expansion property of the
reproduction and selection phases.  Given two states of the RSM
process, the reproduction phase amplifies the distance between them,
and the nature of the selection phase tends to keep this distance
intact.  In this setting, the chain can be described as a noisy random
walk on a boolean hypercube, and our bound reflects the intuition that
when the noise is small, the fast mixing property of the hypercube
should able to enforce fast mixing of the RSM chain.  We reemphasize
that we consider that achieving a better understanding of the mixing
properties of the RSM walk, in terms of both upper and lower bounds
appears to be an interesting and challenging open problem.
\subsection{Proof of Theorem \ref{thm:projected-chains}}
\label{sec:proof-theor-refthm:p}

We now proceed to prove Theorem \ref{thm:projected-chains}.  We will
use notation similar to that used in the proof of Lemma
\ref{lem:proj-strong} in Appendix \ref{sec:proofs-omitted-from-1}.  By
a slight abuse of notation, we will use $P$ and $R$ for the transition
matrices of the projected versions of the chains described in the
proof of Lemma \ref{lem:proj-strong} above.  We note that for any
numbers $M$ and $ M'_1,M'_2,\ldots, M'_l$ summing up to $M$,
$\binom{M}{M'_1,M'_2,\ldots, M'_l}$ is of size at most $O(M\log M)$
bits and can be computed in at most $O(M)$ operations over numbers of
size at most $O(M\log M)$.  We start with an estimation of the time
required to compute the entries of $P$.  Using the above observation,
and the form of entries of $P$, we have the following observation for
$P$.
\begin{observation}
  Each entry of $P$ is of size at most $\tilde{O}(N\max a_\sigma)$,
  and can be computed in $N\max a_\sigma$ operations over
  integers of size $\tilde{O}(N\max a_\sigma)$.
\end{observation}
We now proceed to estimate the complexity of computing entries of the
matrix $R$.  We first get a bound on the time required to pre-compute
the $(L+1)\times (L+1)$ matrix $\mathcal{Q}$ defined in the
description of the mutation chain in Appendix
\ref{sec:proofs-omitted-from-1}.

\begin{observation}
  Let $b$ be the number of bits required to represent $\mu$.  We have
$$
\mathcal{Q}_{st} = (1-\mu)^L\inp{\frac{\mu}{1-\mu}}^{s-t}\sum_{x}\binom{L-s}{x}\binom{s}{t-x}\inp{\frac{\mu}{1-\mu}}^{2x}.
$$
Hence, each entry of $\mathcal{Q}$ is of size at most $\tilde{O}(bL)$ and can
be pre-computed in time $O(L^2)$ operations over numbers of size $\tilde{O}(bL)$. 
\end{observation}

\begin{proof}[Proof of Theorem \ref{thm:projected-chains}]
  Notice that the number of terms in the sum in equation (\ref{eq:6})
  for the computation of the entry of the matrix $R([h], [h'])$ is at
  most
  \begin{eqnarray}
    \prod_{i=0}^{L}\binom{h(i) + L}{L} &\leq&
    \inp{\frac{e}{L}}^{L(L+1)}\prod_{i=0}^{L}(h(i) + L)^L\nonumber\\
    &\leq& \inp{e\inp{1 + \frac{N}{L(L+1)}}}^{L(L+1)}\text{,
      by the AM-GM inequality.}\nonumber
  \end{eqnarray}
  We will use the shorthands $S = \binom{N+L}{L} \leq \frac{(N+L)^L}{L!}$ and $$G =
  \inp{e\inp{1 + \frac{N}{L(L+1)}}}^{L(L+1)}.$$  Notice that $R$ is of
  dimension at most $S\times S$.  Computing the products of all the
  $\mathcal{Q}_{ij}$'s in each of these terms takes $N$
  multiplications on numbers of size at most $\tilde{O}(bL)$, and hence
  produces a number of size $\tilde{O}(NbL)$ in time at most $\tilde{O}(NbL)$.
  Computing the products of all the multinomial coefficients in each
  of the terms takes $O(N)$ multiplications on integers of size $\log
  N$, thus producing an integer of size $O(N\log N)$ in time
  $\tilde{O}(N\log N)$.  The total size of each entry of $R$ is thus
  at most $s_1= \tilde{O}\inp{\log G + N\log N + NbL}$, and
  the entry can be computed in time $t_1 =
  \tilde{O}(Gs_1)$.  All the entries of $R$
  can thus be computed in time $S^2t_1$.

  Notice that $\M_w = PR$, and given the above estimates on the sizes
  of the entries of $P$ and $R$ and the times required to compute
  their entries, each entry of $\M_w$ is of size at most $s_2 = O(\log S + s_1
  + N\max a_\sigma)$, and can be computed in time $\tilde{O}(Ss_2)$,
  given the entries of $P$ and $R$.  Thus, the time taken for the
  computation of $\M_w$, including the computation of $R$ and $P$ and
  the pre-computation of $\mathcal{Q}$ is
  \begin{eqnarray}
    T_1 &=& \tilde{O}\inp{bL^3 + S^3s_2 + S^2Gs_1}\nonumber
  \end{eqnarray}
and each entry of $\M_w$ is of size $$s_2 = \tilde{O}(NbL + N\max
a_\sigma + L^2).$$
The time required to compute an exact solution of $\M_w\pi_w = \pi_w$
with the restriction $\norm[1]{\pi_w} = 1$ using Gaussian elimination
is thus of the order $\tilde{O}(S^3s_2)$.  Comparing this with $T_1$,
we get that the total running time is $\tilde{O}(T_1)$, which is what
we sought to prove.
\end{proof}
\subsection{Proof of Theorem \ref{thm:algor-estim-error}}

\begin{figure*}[htb]
  	\begin{tabularx}{\textwidth}{|X|}
    \hline
	 \vspace{1mm}

         {\bf \textsc{Input:}} An initial state $\alpha_{0} \in
         \Omega_w,$ an $\eps >0$, grid resolution $\delta$, accuracy
         parameter $\delta_1$, and error probability $\delta_2$, $L,N$
         and an $A$ which is class invariant.
  
         \vspace{1mm}
         {\bf \textsc{Goal:}} To estimate $\mu^{h}_{c}(\eps, N).$ 

         \vspace{1mm} 

         {\bf \textsc{Output:}} $\mu_0$ such that$\mu_0 \geq
         \mu_c^h(\epsilon + \delta_1/2)$ and $d^h(\vec{D}_{\infty, \mu_0
           - \delta}, \vec{U}) \geq \epsilon
         -\delta_1/2$. \vspace{1mm}

         Let $\eps'= \frac{\delta_1}{2L^2}$ (the distance from
         stationarity), $c = \ceil{\frac{1}{2\delta}}$  and $s =
         \ceil{\frac{8L^4}{\delta_1^2}\log(2c(L+1)/\delta_2)}$ (number of
         samples from the distribution).
	
	\vspace{1mm}
	{\bf For} $\delta \leq \mu  \leq 1/2$ in steps of $\delta$,
	\begin{enumerate}
	
	\item  Let $\tau \defeq \tau(L,N,A,\mu,\eps').$
		
	\item Let $\vec{D}_{\tau,k},$ for $k=1,\ldots,s,$ denote $s$
          independent samples from the RSM process with parameters
          $L,N,A$ and $\mu$ starting from the initial state
          $\alpha_{0}.$
	
	\item Let $\vec{Z}_{s} \defeq  \frac{1}{s} \sum_{k=1}^{s} \vec{D}_{\tau,k}.$
	
	\item {\bf If} $d^h (\vec{Z}_{s},U) \leq \eps$ then {\bf
            Return} $\mu$ and {\bf Stop}.
	
	\item {\bf Else} $\mu = \mu + \delta.$ 
	
	\end{enumerate}
\\
\hline 
\end{tabularx}
  \caption{The  Algorithm {\sc ErrorThreshold}}
  \label{fig:algorithm}
\end{figure*}

\label{sec:proof-theor-refthm}
In this section, we give a proof of Theorem
\ref{thm:algor-estim-error}.  Recall that $\pi_w$ denotes the
stationary distribution of the projected RSM chain described in
Section \ref{sec:prel-defin}.  Using the definition of the mixing time,
and the Chernoff-Hoeffding bound, we  have the following lemma in
order to bound the number of samples $s$ required in each iteration of
the algorithm:
\begin{lemma}
  Suppose we take $s$ samples from independent realizations of the
  fully mixed projected RSM process, denoting the samples so obtained
  as $\vec{D}(i)$ for $i = 1,2,\ldots, s$.  For any fixed genome
  $\sigma$, we then have $\E{D^\sigma(i)} = \Ex{\pi_w}{D^\sigma}$ for
  all $i$.  Now for every $\epsilon > 0$
$$
\Pr{\abs{\frac{1}{s}\sum_{i=1}^sD^\sigma(i) - \Ex{\pi_w}{D^\sigma}} \geq
  \epsilon} \leq 2\exp\inp{-2\epsilon^2s}.
$$
In particular, if we run $s$ independent realizations of the chain up
to the time $\tau(L, N, A, \mu, \epsilon/2)$, and let $\vec{D}_{i}$
denote the sample obtained by the $i$th realization, then
$$
\Pr{\abs{\frac{1}{s}\sum_{i=1}^sD^\sigma(i) - \Ex{\pi_w}{D^\sigma}} \geq
  \epsilon} \leq 2\exp\inp{-\epsilon^2s/2}.
$$
\end{lemma}

\begin{proof}[Proof of Theorem \ref{thm:algor-estim-error}]
  Consider the random variables $\vec{Z_s}$ in any of the at most $c$
  iterations in \textsc{ErrorThreshold}.  By the choice of $\epsilon'$
  and $s$, we get from the above lemma and a union bound that all the
  $c$ different variables $\vec{Z_s}$ that we get across all the iterations of the
  loop satisfy
  \begin{equation}
    \label{eq:8}
    \norm[\infty]{\vec{Z_s} - \Ex{\pi_w}{\vec{D}}} \leq \delta_1/(2L^2)
  \end{equation}
  with probability at least $1-\delta_2$.  In the rest of the proof,
  we condition on the above event occurring. 

  Since the average Hamming distance can be at most $L$, and we are
  running the projected chains till $\tau(L, N, A, \mu, \delta_1/(2L^2))$
  for each $\mu$, we get that for every $\mu$ considered by the
  algorithm
$$
\abs{d^h(\pi_{w,\mu}, \vec{U}) - d^h(\vec{Z_s}, \vec{U})} \leq
\delta_1/2.
$$
We therefore deduce that when the algorithm outputs a $\mu$,
$d^h(\pi_{w,\mu}, \vec{U}) \leq \epsilon + \delta_1/2$, using the
conditioning on the event in equation (\ref{eq:8}).  Similarly, by
noticing that the algorithm had $d^h(\vec{Z_s}, \vec{U}) \geq
\epsilon$ for $\mu - \delta$, we get that $d^h(\pi_{w,\mu-\delta},
\vec{U}) \geq \epsilon - \delta_1/2$.  The estimate of the running
time follows by noticing that assuming bits with bias $\mu$ can be
sampled in $O(\log(1/\mu))$ time, it takes time $O(NL\max
a_\sigma\log(1/\mu))$ to simulate each step of the $\M_w$ chain.  The
bound on the error probability follows from the conditioning used on
the validity of equation (\ref{eq:8}).
\end{proof}

We comment briefly on a technical point about the definition of the
error threshold that has been used in the literature (and that we use
too).  With this definition, there might exist a $\mu$ satisfying
$\nfrac{1}{2} > \mu > \mu_c^h(\epsilon)$ such that $d^h(\vec{v},
\vec{U}) > \epsilon$.  An analogous condition might hold in the finite
population case too.  If we could preclude the occurrence of such
anomalous behavior of error-thresholds, we would be able to improve
the guarantee on the output $\mu_0$ of the algorithm {\sc
  ErrorThreshold} to be of the form $\mu_c^h(\epsilon + \delta_1, N)
\leq \mu_0 \leq \mu_c^h(\epsilon - \delta_1) + \delta$.  We observe,
however, that somewhat surprisingly, even for the simpler case of the
quasispecies model, to the best of our knowledge, no attempts have
been made to rigorously prove that such anomalous behavior cannot
occur.  We leave the resolution of this point for the finite
population case as an open problem.

\end{document}